\documentclass[reqno]{amsart}
\usepackage[english]{babel}
\usepackage[T1]{fontenc}
\usepackage[utf8]{inputenc}
\usepackage{amsmath}
\usepackage{amsfonts}
\usepackage{array}
\usepackage{eurosym}
\usepackage{cite}
\usepackage{upgreek}

\usepackage{amssymb}
\usepackage{mathrsfs}
\usepackage{amsthm}
\usepackage{xfrac}
\usepackage{lmodern}
\usepackage{fix-cm}
\usepackage{listings}
\usepackage{fancyvrb}
\usepackage{enumerate}   
\usepackage{xcolor}
\usepackage{pgf,tikz}
\usepackage{mathrsfs}
\usetikzlibrary{arrows}
\usepackage{float}
\usepackage{subcaption}
\usepackage{hyperref}{}
\usetikzlibrary{patterns}
\usepackage{soul}

\usepackage[pagewise]{lineno}

\usepackage{wasysym}
\usepackage{enumitem}

\newtheorem{thm}{Theorem}[section]

\newtheorem{lem}[thm]{Lemma}

\theoremstyle{definition}
\newtheorem{defn}[thm]{Definition}

\newtheorem{remark}[thm]{Remark}
\newtheorem*{ass}{Assumption}

\newtheorem*{notation}{Notation}

\def\Z{\mathbb Z}
\def\R{\mathbb R}
\def\N{\mathbb N}
\def\G{\mathbb G}
\def\V{\mathbb V}
\def\A{\mathbb A}

\def\cA{\mathcal A}
\def\zu{\{0,1\}}
\newcommand\blfootnote[1]{%
  \begingroup
  \renewcommand\thefootnote{}\footnote{#1}%
  \addtocounter{footnote}{-1}%
  \endgroup
}

\title[A survey on Functional Encryption]{A survey on Functional Encryption}
\author{Carla Mascia}
\author{Massimiliano Sala}
\author{Irene Villa}
\address{University of Trento}
\email{$\{$carla.mascia, massimiliano.sala, irene.villa$\}$@unitn.it}

\begin{document}
\maketitle 
\blfootnote{ Carla Mascia, Massimiliano Sala, Irene Villa. \emph{A survey on Functional Encryption}. Advances in Mathematics of Communications, doi: 10.3934/amc.2021049 (2021), \url{https://www.aimsciences.org/article/doi/10.3934/amc.2021049}}

\begin{abstract}
Functional Encryption (FE) expands traditional public-key encryption in two different ways: it supports fine-grained access control and allows learning a function of the encrypted data.
In this paper, we review all FE classes, describing their functionalities and main characteristics.
In particular, we mention several schemes for each class, providing their security assumptions and comparing their properties.
 To our knowledge, this is the first survey that encompasses the entire FE family.

\end{abstract}

\section{Introduction}
Recent emerging applications, such as cloud services, highlight the need for a more complex notion of encryption than public-key encryption.
Standard public-key encryption has an  ``all-or-nothing'' access to encrypted data; that is, either a user can decrypt and read a secret message (plaintext), or the user will not get any further information about it in addition to what they already know from the encrypted message (ciphertext).
Nevertheless, in some situations, we may want the receiver to know only some partial information of the plaintext.
Or we may restrict the decryption of ciphertext to individuals who satisfy specific requirements, whereas traditional public-key encryption provides only coarse-grained access to encrypted data, where a single private key can decrypt all data.

In 1984, Shamir \cite{Sha84} proposes the concept of \textit{Identity-Based cryptography} as a way to simplify certificate management in email systems.
Users' identifier information, such as email addresses, can be used as public keys for encryption, so it is possible to encrypt a message to any user whose email address is known.
The receiver requests a decryption key (related to its identity) to an authority or a third party and then decrypts the received ciphertext.
The first practical realisations are presented in 2001, by Boneh and Franklin \cite{BF01}, and by Cocks \cite{Coc01}. 
In 2005, Sahai and Waters \cite{SW05} introduce the \textit{Fuzzy Identity-Based Encryption}, a generalisation of the system proposed by Shamir.
The following year, Goyal, Pandey, Sahai, and Waters \cite{GPSW06} propose the more general class of \textit{Attribute-Based Encryption}, where the access to encrypted data is given only to users with specific attributes or policies. This is the first realisation of fine-grained access control over a cryptographic system.

Later in 2007, Boneh and Waters \cite{BW07} propose the so-called class of \textit{Hidden Vector Encryption}, allowing comparison and subset queries.
Katz, Sahai, and Waters \cite{KSW08} in 2008 present an \textit{Inner-Product Predicate Encryption} scheme, allowing more complex queries such as evaluation of polynomials.
Finally, in 2010 O'Neill \cite{ON10} and in 2011 Boneh, Sahai and Waters \cite{BSW11} introduce the concept of \textit{Functional Encryption}, shortly FE. This new broad vision of encryption systems comprehends all the systems as mentioned above.
FE  expands the functionalities of traditional public-key encryption in two different ways.
First, as we have seen, it supports access control policies, allowing only users who satisfy a certain procedure to decrypt the secret information.
Secondly, it permits selective computations on the ciphertexts, allowing users to learn a function of the encrypted data. 
In 2015, Abdalla, Bourse, De Caro, and Pointcheval \cite{ABDP15} present the first scheme of FE that outputs a functionality of the encrypted message: \textit{Inner Product Encryption}.

In an FE scheme, there are different keys with different roles. In general, a {\em public key} is generated to allow anyone in the system to encrypt any private message $m$. {\em Functional decryption keys}, also referred to as {\em secret keys}, are generated from a {\em master key}, which is supposed to be secret and kept by an authority. Each decryption key is associated with a specific function $f$.
Decrypting an encryption of a message $m$ with a secret key associated to a function $f$ reveals directly the value $f(m)$. 
There also exists a different scenario, called \textit{private-key setting}, where the key used to encrypt a message is not publicly known.
The security of an FE scheme is based on the following: even though an adversary has access to functional decryption keys for different functionalities, it will not learn anything more than what each individual functional decryption key allows to know.

Due to its expressive access control, FE has many applications in practical situations, such as large organisations, where data is often shared according to some access policy.
FE also comprehends \textit{Searchable Encryption}, one of the active research trends in cryptography.
Searchable encryption, also known as \textit{Public-key Encryption with Keyword Search} \cite{BCOP04}, allows searching a piece of information over encrypted data while keeping the data content and the query secret even from the server holding the encrypted data.
When storing sensitive data on remote servers, this guarantees the secrecy of users' sensitive data while preserving searchability on the server-side.
FE can also be used for mining large datasets.
It is beneficial when data is highly sensitive, for example, medical data, but revealing some partial or aggregated information on this data does not compromise the privacy of the ``data-owner''. This can be exploited, for example, to perform statistical analysis for medical research.

As mentioned, one of the fundamental ideas behind FE is to learn a function of the original message. Another cryptographic notion which achieves a similar result is \textit{Fully Homomorphic Encryption} (FHE) \cite{Gen09}. Intriguingly, FE and FHE seem somehow intertwined, as recognised in \cite{GGH+16,alwen2013relationship}.
A notable difference between FE and FHE is the way the functions are applied. FE grants control over which functions can be applied to the data via a master-key holder, who issues secret keys based on the appropriateness of the function. The user then uses the secret key to obtain a function of the plaintext in the clear. Whereas FHE permits functions to be run by anyone with the so-called \textit{evaluation key}, but the result is still encryption. Therefore, only the owner of the secret key can decrypt the result. \\
Yet, the concepts of FHE and FE do have some overlap, and it has been demonstrated that functional encryption might work as FHE, with some slight adaptation \cite{alwen2013relationship}.

\subsection*{Our contribution}
In this paper, we survey the families of functional encryption schemes.
Our work aims to give the reader a wide overview of FE schemes and provide a satisfying summary of state-of-art. 
In particular, we focus on the functionalities of different FE schemes describing their limitations, their security models, and the mathematical assumptions involved.
We also present proof of some implications among these assumptions.
Given the vast literature, we do not provide detailed descriptions of the existing constructions of FE schemes; instead, we illustrate here the main ideas upon which these schemes are based, and we provide some of the most relevant and recent works.
In addition, we present an overview of ``non-standard'' FE schemes, that is, FE schemes that go beyond the inner product encryption and schemes with enhanced properties.
Indeed, in real scenarios, some characteristics of ``standard'' FE are not desirable, such as having unique authority. Hence schemes with more complex features have been studied and proposed, such as decentralised or multi-authority schemes.

\subsection*{Organization} 

The survey is organized as follows.
 Sections \ref{chap: def FE} and \ref{chap : math_ass} are both devoted to introducing preliminaries. In particular, first, we give the definition of Functional Encryption and Predicate Encryption, the latter being a large class of FE. Then, we provide a systematic presentation of the main mathematical assumptions used to prove the security of FE schemes, underlining the relations between some of them. 
 Section \ref{chap: security} briefly presents the two most used models of security for an FE scheme: indistinguishability-based security and simulation-based security. 
 Section \ref{chap: IBEandABE} gives an overview of the well-known classes of Identity-Based Encryption and Attribute-Based Encryption.
  Section \ref{chap: PE} focuses on the class of Predicate Encryption  and
  Section \ref{chap: IPE} describes the class of Inner Product Encryption.
  In Section \ref{chap: other}, we present many different generalisations of the previous schemes.
  Finally, in Section \ref{chap: concl} we draw our conclusions.

\section{Notations and Definition of FE}\label{chap: def FE}

In this section, we collect all the notations and preliminaries needed for the rest of this work. 

When we talk about algorithms or adversaries, we refer to Turing machines.
The acronym PPT stands for \textit{Probabilistic Polynomial Time}, so a PPT algorithm is a probabilistic Turing machine, equipped with an extra randomness tape, that solves problems in polynomial time.

Let $n \in \N$, $\Z_n$ stands for the ring of integers modulo $n$ and $\Z_n^* = \Z_n \setminus \{0\}$.

\begin{defn}
A function $f:\N\rightarrow\R$ is called \textit{negligible} if for any positive integer $d$ there exists an integer $N_d$ such that $|f(k)|<\frac{1}{k^d}$ for any $k\ge N_d$.  
\end{defn}

To define a functional encryption scheme, we follow the definition given in \cite{BSW11}, which is the first attempt to formalise FE.
\begin{defn}\label{def: FE scheme} Let $K$ be the key space containing a special key $\epsilon$, called the \textit{empty key}, and $X$ be the plaintext space. A \textit{functionality} $F$ defined over $K \times X$ is a function $F:K\times X\rightarrow Y$, where $Y$ is the output space containing a ``fail'' character $\perp$. 
A\textit{ functional encryption scheme} for the functionality $F$  is a tuple of four PPT algorithms (\texttt{Setup}, \texttt{KeyGen}, \texttt{Enc}, \texttt{Dec}) satisfying the following conditions for all $k\in K$ and all $x\in X$:\\

\begin{tabular}{l l}
$\texttt{Setup}(1^\lambda)\rightarrow({\rm pp},{\rm mk})$ & on input the security parameter $\lambda$, outputs a public\\
& key $\rm pp$ and a private key $\rm mk$, called \textit{master key};
\\
$\texttt{KeyGen}(\mathrm{mk},k)\rightarrow {\rm sk}_k$ & on input the master key $\rm mk$ and a key $k\in K$,\\
& outputs a secret key ${\rm sk}_k$;
\\
$\texttt{Enc}({\rm pp},x)\rightarrow c$ & on input the public key $\rm pp$ and a message $x\in X$,\\
& outputs a ciphertext $c$;
\\
$\texttt{Dec}({\rm sk}_k,c)\rightarrow y$ & on input a secret key ${\rm sk}_k$ and a ciphertext $c$,\\
& outputs $y$
\end{tabular}

where we require that $y=F(k,x)$ with probability 1.
\end{defn}
The special key $\epsilon$ captures all the information about the plaintext that intentionally leaks from the ciphertext.  Thus, if a user runs $\texttt{Dec}(\epsilon,c)$ on a ciphertext $c=\texttt{Enc}({\rm pp},x)$, it will obtain all the information about $x$ that intentionally are revealed from $c$.  For example, if \texttt{Enc} outputs a ciphertext of the same length as the plaintext, then the encrypted message intentionally reveals $\texttt{len}(x)$, namely the size of the plaintext.

\begin{remark}
Standard public-key encryption is the simplest example of functional encryption.
With $K=\{1,\epsilon\}$, for $k=1$ we get the decryption of a valid ciphertext, that is,
$$F(k,x)=\left\{\begin{array}{ll}
x &\mbox{if }k=1,\\
\texttt{len}(x) &\mbox{if }k=\epsilon.
\end{array}\right. $$
\end{remark}
\begin{remark}\label{rmk:k->f}
Given a functionality $F$, let $\mathcal{F} = \{f_k \ | \ k \in K \text{ and } f_k(\cdot) = F(k, \cdot )\}$ be a family of functions derived from $F$. An FE scheme can be defined starting from $\mathcal{F}$, instead of $F$, and the resulting scheme is the same as that presented in Definition \ref{def: FE scheme}. In an FE scheme for $\mathcal{F}$, the \texttt{Setup} and \texttt{Enc} algorithms are defined as in Definition \ref{def: FE scheme}, whereas \texttt{KeyGen} takes in input the master key and a function $f \in \mathcal{F}$, it outputs a secret key $\mathrm{sk}_f$. The decryption algorithm, using $\mathrm{sk}_f$ on the encryption of $x$ returns $f(x)$, that is, a function of the plaintext $x$.
\end{remark}
\begin{remark}
For some practical applications, it is relevant to distinguish between \textit{public-key} and \textit{private-key setting}. In the former, anyone can encrypt using the public key generated in the \texttt{Setup} phase; in the latter, only those who own the master key can encrypt. 
Notice that the definition of FE scheme in Definition \ref{def: FE scheme} is given in the public-key setting.
\end{remark}

Regarding the notion of security for an FE scheme, we dedicate a separate section, see Section \ref{chap: security}.

Many of the FE schemes present in the literature belong
to the class of \textit{predicate encryption}, a powerful tool that enables fine-grained access to encrypted information such as searching on encrypted data. Regarding the notion of security for an FE scheme, we dedicate a separate section, see Section \ref{chap: security}.

\begin{defn}
A \textit{predicate} is a Boolean-valued function $f: A \rightarrow \{0,1\}$, where $A$ is an arbitrary set, and the elements $0$ and $1$ of the codomain are interpreted as the logical values \emph{false} and \emph{true}, respectively.
\end{defn}
In many applications, a plaintext $x\in X$ is a pair $({\rm ind},m)\in I\times M$, where $I$ is called \textit{index set} and $M$ is called \textit{payload set}.
A \textit{Predicate Encryption} (PE) scheme is then defined in terms of a polynomial-time predicate $P:K\times I\rightarrow\zu$ and the functionality over $K\times (I\times M)$ is defined as
$$F(k\in K\setminus\{\epsilon\}, ({\rm ind},m) )=\begin{cases}
m & \mbox{if } P(k,{\rm ind})=1,\\
\perp & \mbox{if } P(k,{\rm ind})=0.
\end{cases}$$

FE schemes can be divided into three large classes, depending on the features and functionalities they achieve:\begin{itemize}
\item Predicate encryption systems with public index;
\item Predicate encryption systems with private index;
\item Functional encryption systems beyond predicates.
\end{itemize}

In the public-index class, the plaintext index is easily readable from the ciphertext, formally by evaluating $F$ in the empty key $\epsilon$.  
And among these schemes, one finds  Identity-Based Encryption (IBE) and Attribute-Based Encryption (see Section \ref{chap: IBEandABE}). We refer to the schemes with private index simply with Predicate Encryption (PE), and among them, one finds Anonymous IBE, Hidden Vector Encryption (HVE), and Inner Product Predicate Encryption (see Section \ref{chap: PE}). Finally, the last class includes Inner Product Encryption (see Section \ref{chap: IPE}).

\def\cG{\mathcal G}
\def\cB{\mathcal B}
\def\cF{\mathcal F}

\section{Mathematical background and  assumptions}\label{chap : math_ass}
This section presents some of the main assumptions used to prove the security of FE schemes mentioned in this work.

In general, a security assumption is based on the hardness of solving a specific problem.
With \textit{hardness} we mean that an adversary (which is a PPT algorithm) can solve the problem only with a negligible probability.
We will also use \textit{intractable}, \textit{difficult}, or \textit{challenging} as synonyms for hard.

In this paper, $q$ denotes a prime number, with $q \geq 5$.

\subsection{Bilinear Groups} 
In this subsection, we describe some problems and assumptions related to bilinear groups studied and used in FE schemes. Some of the problems presented in the following are formulated in the symmetric case, the others in the more general asymmetric setting.\\

\begin{defn}\label{def:BilGrp}
A \textit{bilinear group} is a tuple $(q,\G_1,\G_2,\G_T,e)$ such that:
\begin{itemize}
\item  $\G_1,\G_2,\G_T$ are (multiplicative) cyclic groups of prime order $q$ where the group actions can be computed efficiently;
\item $e:\G_1\times\G_2\rightarrow\G_T$ is a \textit{bilinear} map, that is, $e(g_1^a,g_2^b)=e(g_1,g_2)^{ab}$, for all $g_1,g_2$ elements in $\G_1,\G_2$, respectively, and all $a,b$ in $\Z_q$;
\item $e$ is \textit{admissible}, that is \begin{itemize}
\item $e$ is efficiently computable for any pair in $\G_1\times\G_2$;
\item $e$ is non-degenerate: the map does not send all pairs in $\G_1\times\G_2$ to the identity $1_{\G_T}$ in $\G_T$ (i.e., $e(g_1,g_2)\ne1_{\G_T}$, when $g_1 \neq 1_{\G_1}$ and $g_2 \neq 1_{\G_2}$).
\end{itemize}
\end{itemize}
The function $e$ is simply called \textit{bilinear map}, or \textit{pairing}. When $\G_2=\G_1$ we talk about \textit{symmetric bilinear group}.
In this case, we identify the bilinear group with $(q,\G_1,\G_T,e)$.
\end{defn}

We recall  that, since $q$ is a prime, every element in $\G_i$ (excluding $1_{\G_i}$) is a generator of $\G_i$, for $i=1,2,T$.

\begin{remark}
It is also possible to consider the groups $\G_1$ and $\G_2$ as additive cyclic groups. In this case, the bilinearity property corresponds to $e(aP,bQ)=e(P,Q)^{ab}$, for $P\in\G_1$ and $Q\in\G_2$. For the assumptions presented in the following, there is an equivalent one with $\G_1$ and $\G_2$ additive groups. For the sake of simplicity, we describe here only the assumptions in the multiplicative form. 
\end{remark}
Explicit constructions of admissible symmetric bilinear maps have been presented, see, e.g., \ \cite{BF01} where concrete examples of groups $\G_1,\G_T$ are given, together with an admissible bilinear map between them. 
Unfortunately, computing pairings incurs a high computational cost and represents the bottleneck of using pairings in actual protocols. In general, $\G_1$ is an abelian variety over some fields. For instance, $\G_1$ is typically the group of points on an elliptic curve over a finite field. The choice of the elliptic curve, e.g., \ a supersingular curve or an MNT curve \cite{MNT01}, has performance implications. Whereas $\G_T$ usually is a finite field. The most common examples of bilinear maps are the (modified) Weil pairing and Tate pairing of supersingular elliptic curves. These maps are non-trivial to compute, and the known algorithms used to compute them, e.g., \ the Miller’s algorithm \cite{miller2004weil}, and refinements thereof \cite{blake2006refinements},  are computationally expensive. 
For an overview of elliptic curves, we refer the interested reader to \cite{Sil09, page2006comparison}. 

\begin{notation}
In the following, $g$ and $h$ denote random generators of $\G_1$ and $\G_2$, respectively. 
\end{notation}

We are ready to introduce all the necessary assumptions.

\begin{ass}[CDH]
The \textit{Computational Diffie-Hellman}  problem in $\G_1$ is to compute the element $g^{ab}$ given the tuple $\left[ g,g^a,g^b\right]$ for random $a,b\in\Z_q$.
The CDH assumption holds if the CDH problem is hard to solve.
\end{ass}

\begin{ass}[DDH]
The \textit{Decisional Diffie-Hellman}  problem in $\G_1$ is, given the tuple $\left[ g, g^a,g^b\right]$, to distinguish between $g^{ab}$ and $ g^c$,
where $a,b,c$ are random in $\Z_q$.
The DDH assumption holds if solving the DDH problem is difficult.
\end{ass}

In this section, $k$ denotes a positive natural number.

\begin{ass}[PDDH]
Let $k \geq 2$. The \textit{$k$-Party Decisional Diffie-Hellman}  problem in $\G_1$ is, given the tuple $\left[g, g^{a_1},\dots,g^{a_k}\right]$, to distinguish between $g^{a_1 \cdots a_k}$ and $g^{b}$, where $a_1, \dots, a_k, b$ are random in $\Z_q$.
The $k$-PDDH assumption holds if solving the $k$-PDDH problem is hard.
\end{ass}
Note that 2-PDDH assumption is equal to the DDH assumption. The 3-PDDH is first introduced in \cite{boneh2006fully}.

Many cryptographic schemes are based on the \textit{discrete logarithm problem}, that is, given  $g_1,g_2\in\G$ for a finite group $\G$, where $g_2=g_1^a$ for some $a \in \mathbb{Z}_{|G|}$, to find $a$.
Bilinear maps can provide a connection between the discrete logarithm problems of different groups.
Moreover, the DDH problem is affected by these maps.
The following result is well-known, see for instance \cite{BF01}.
\begin{lem}
If there exists an admissible bilinear map $e:\G_1\times\G_1\rightarrow\G_T$, then 
\begin{enumerate}
    \item the DDH problem in  $\G_1$ is easy to solve;
    \item the discrete logarithm problem in $\G_1$ can be reduced to the discrete logarithm problem in $\G_T$ (MOV reduction \cite{MOV93}). 
\end{enumerate}
\end{lem}

\begin{proof}
\ 
\begin{enumerate}
    \item Given $\left[ g,g^a, g^b\right]$ and a fourth element $g^c$, to verify whether $c$ is a random element in $\Z_q$ or $c=ab$, it is sufficient to compute and compare $e(g^a,g^b)=e(g^{ab},g)$ and $e(g^c,g)$.
    \item Given $\left[g,g^a\right]$ in $\G_1$, one can compute $e(g, g^a) = e(g,g)^{a}$.  Since $e(g,g)$ is not the identity of $\G_T$ due to the non-degeneracy of $e$, if an adversary can solve the discrete logarithm problem in $\G_T$, it is able to find $a$.
\end{enumerate}
\end{proof}
Note that even assuming the existence of a bilinear map from $\G_1$ to $\G_T$, the CDH problem in $\G_1$ can still be hard.

Since the DDH problem can be easily solved with a bilinear map, several variants of this assumption have been proposed, and many cryptosystems base their security on these variants.

\begin{ass}[XDH and SXDH]
Let $(q,\G_1,\G_2,\G_T,e)$ be an asymmetric bilinear group such that the CDH problem is intractable in both $\G_1$ and $\G_2$.
The \textit{external Diffie-Hellman} (XDH) assumption states that the DDH problem is  intractable in $\G_1$. 
If the DDH problem is also intractable in $\G_2$, we have the \textit{symmetric external Diffie-Hellman} (SXDH) assumption.
\end{ass}
\begin{remark}
 In practice, to our knowledge, there is no group in which the CDH problem is intractable, but the DDH problem is not.
\end{remark}

\begin{ass}[BDH]
The (symmetric and computational) \textit{Bilinear Diffie-Hellman}  problem for $(q,\G_1,\G_T,e)$ is to compute $e(g,g)^{abc}$ given $\left[ g,g^a,g^b,g^c\right]$ for random $a,b,c\in\Z_q$.
The BDH assumption holds if the BDH problem is challenging to solve.
\end{ass}

\begin{remark}
The BDH problem in $(q,\G_1,\G_T,e)$ is no harder than the CDH problem in $\G_1$ or $\G_T$.
That is, an algorithm that can solve the CDH problem in $\G_1$ or $\G_T$ is sufficient for solving BDH in $(q,\G_1,\G_T,e)$.
\end{remark}

\begin{ass}[DBDH]
Consider now $a,b,c,x\in\Z_q$ taken at random. The \textit{Decisional Bilinear Diffie-Hellman}  assumption in  $(q,\G_1,\G_T,e)$ states that, given the tuple $\left[ g,g^a,g^b,g^c\right]$, it is hard to distinguish $e(g,g)^{abc}$ from a random element of $\G_T$.\end{ass}
In \cite{BW06}, we find the DBDH problem for $(q,\G_1,\G_2,\G_T,e)$ in its asymmetric version: to distinguish between $e(g,h)^{abc}$ and a random element in $\G_T$, given a tuple $\left[ g,g^a,g^c,h,h^a,h^b\right]$ for random $a,b,c\in\Z_q$ and $g,h$ in $\G_1,\G_2$, respectively.

In \cite{SW05}, the \textit{Decisional Modified Bilinear Diffie-Hellman} assumption is described.
\begin{ass}[DMBDH]
The \textit{Decisional MBDH} assumption in $(q,\G_1,\G_T,e)$ holds if, given $\left[ g,g^a,g^b,g^c\right]$, it is hard to distinguish $e(g,g)^\frac{ab}{c}$ from $e(g,g)^z$, for random $a,b,c,z\in\Z_q$. 
\end{ass}

Another modification of the DBDH assumption is the \textit{Decisional Bilinear Diffie-Hellman Inversion} assumption that we present below together with the corresponding computational version.
\begin{ass}[DBDHI]
The $k$-BDHI problem for $(q,\G_1,\G_T,e)$ is defined as follows:
given the tuple $\left[ g,g^x,g^{x^2},\ldots,g^{x^k}\right]$ for random $x\in\Z_q$, to compute $e(g,g)^{1/x}$.
The BDHI assumption holds if this problem is intractable.
The \textit{decisional $k$-BDHI} problem   
requires then to distinguish $e(g,g)^{1/x}$ from a random element in $\G_T$. The decisional $k$-BDHI assumption holds if this problem is challenging.
\end{ass}

The following result is well-known (see for example \cite{BB04}), but for the sake of completeness, we present proof.

\begin{lem}
The 1-BDHI assumption is equivalent to the standard BDH assumption.
\end{lem}
\begin{proof}
Let $\mathcal{A}$ be an adversary. Assume that $\mathcal{A}$ can solve 
the BDH problem and it is given $\left[ g,g^x\right]$. $\mathcal{A}$ sets $t=g^x$ and $a=b=c=\frac{1}{x}$ ($t^a=g$). Solving the BDH problem over $\left[ t,t^a,t^b,t^c\right]$ determines $e(t,t)^{abc}$. Since $e(t,t)^{abc}=e(g^x,g^x)^\frac{1}{x^3}=e(g,g)^\frac{1}{x}$, the adversary can solve the  1-BDHI problem.\\
On the other hand, assume now that $\mathcal{A}$ can solve 
 the 1-BDHI problem and takes in input $\left[ g,g^a,g^b,g^c\right]$. 
For $d \in \Z_q$, if $\mathcal{A}$ sets $t=g^d$ then $\left[ g,g^d\right]=\left[ t,t^\frac{1}{d}\right]$ and it can compute $e(t,t)^d=e(g,g)^{d^3}$. Since $\gcd(6,q) = 1$ and
$$abc = \dfrac{1}{6} \left[(a+b+c)^3 - (a+b)^3-(b+c)^3-(a+c)^3+a^3+b^3+c^3\right],$$
$\mathcal{A}$ can compute the term $e(g,g)^{abc}$ from the terms $e(g,g)^{d}$, with $d \in \{a,b,c,a+b,a+c,b+c,a+b+c\}$. It follows that $\mathcal{A}$ can solve the BDH problem, as well.
\end{proof}

\begin{remark}
It is not known whether the $k$-BDHI assumption, for $k \geq 2$, is equivalent to BDH.
Instead, when considering the decisional formulation of the mentioned assumptions, the decisional BDH assumption is no stronger than the decisional 1-BDHI assumption.
Indeed, the first implication in the previous proof can be modified using the decisional formulation without affecting its validity.
\end{remark}

\begin{ass}[DBDHE]
The \textit{Decisional $k$-Bilinear Diffie-Hellman Exponent}  problem for $(q,\G_1,\G_T,e)$ is defined as follows: given a tuple $$\left[ g_1,g_1^a, g_1^{a^2},\ldots,g_1^{a^k},g_1^{a^{k+2}},\ldots,g_1^{a^{2k}},g_2\right], $$
where $a\in\Z_q$ and $g_1,g_2\in\G_1$ are chosen at random, to distinguish $e(g_1,g_2)^{a^{k+1}}$ from a random element in $\G_T$. The decisional $k$-BDHE assumption holds if the decisional $k$-BDHE problem is difficult to solve.
\end{ass}
The following assumptions are both modifications of the DBDHE assumption introduced in \cite{Wat11} and \cite{Gen06}, respectively. 

\begin{ass}[DPBDHE]
The \textit{Decisional Parallel $k$-Bilinear Diffie-Hellman Exponent} problem for $(q,\G_1,\G_T,e)$ is defined as follows.  
Consider $a,s,b_1,\ldots,b_k\in\Z_q$  chosen at random.
The problem consists in, given \begin{align*}
\Big[&g,g^s,g^a,\ldots,g^{a^k},g^{a^{k+2}},\ldots,g^{a^{2k}} \\
&\forall \ 1\le j\le k, \ g^{s b_j},g^{a/b_j},\ldots,g^{a^k/b_j},g^{a^{k+2}/b_j},\ldots,g^{a^{2k}/b_j},\\
&\forall \ 1\le i,j\le k, i\ne j, \  g^{a s b_i/b_j},\ldots,g^{a^k s b_i/b_j}\Big],
\end{align*}
 distinguishing  $e(g,g)^{a^{k+1}s}\in\G_T$ from a random element in $\G_T$. The decisional $k$-PBDHE assumption holds if the decisional $k$-PBDHE problem is hard to solve. 
\end{ass}

\begin{ass}[DABDHE]
The \textit{Decisional $k$-Augmented Bilinear Diffie-Hellman Exponent} problem for $(q,\G_1,\G_T,e)$ consists in, giving the tuple
$$\left[ g_1,g_1^a, g_1^{a^2},\ldots,g_1^{a^k},g_1^{a^{k+2}},\ldots,g_1^{a^{2k}},g_2,g_2^{a^{k+2}}\right], $$
 distinguishing $e(g_1,g_2)^{a^{k+1}}$ from a random element of $\G_T$,
where $a\in\Z_q$ and $g_1,g_2\in\G_1$ are  chosen at random.
The decisional $k$-ABDHE assumption holds if the DABDHE problem is difficult to solve.
Gentry in \cite{Gen06} also presents a truncated version of this problem, where the terms $g_1^{a^{k+2}},\ldots,g_1^{a^{2k}}$ are omitted from the given tuple.
\end{ass}

We now consider other assumptions not directly based on the Diffie-Hellman problem. 
\begin{ass}[DLIN]
As presented in \cite{BBS04}, the  \textit{Decisional Linear}  assumption for $\G_1$  states that, given the tuple $\left[ g, g^a, g^b, g^{ac}, g^{bd} \right]$ for random $a,b,c,d\in\Z_q$, it is hard to distinguish $g^{c+d}$ from a random element of $\G_1$.
\end{ass}

The following claim is well-known, see for example \cite{BW06}, but for sake of completeness we present a proof.
\begin{lem}\label{Lemma: DLIN implies DBDH}
DLIN assumption for $\G_1$ implies DBDH assumption for $(q, \G_1, \G_T, e)$.
\end{lem}

\begin{proof}
Let $\cA$ be an adversary and assume that $\cA$ can solve DBDH. Given the tuple $\left[ g,g^a,g^b,g^{ac},g^{bd}\right]$ the goal is to decide whether an element $g^t$ is $g^{c+d}$ or a random element in $\G_1$.
The adversary first computes $e(g^b,g^{ac})e(g^a,g^{bd})=e(g,g)^{ab(c+d)}$, then $\cA$ considers the DBDH problem on input $\left[ g,g^a,g^b,g^{t}\right]$. Since $\cA$ can distinguish whether $e(g,g)^{ab(c+d)}$ is $e(g,g)^{abt}$ or not, then $\cA$ can distinguish whether $t=c+d$ or not.
\end{proof} 

Notice that this proof cannot be directly extended to the computational cases.

In \cite{BW06}, it is presented an asymmetric version of  DLIN assumption for the bilinear group $(q,\G_1,\G_2,\G_T,e)$, where the given tuple includes also the elements $h,h^a,h^b$.
Note that the implication showed in Lemma \ref{Lemma: DLIN implies DBDH} can be extended similarly to the asymmetric version of the assumptions.

One variant of the DLIN assumption is the \textit{External Decisional Linear} assumption.
Basically, it extends the DLIN problem to an asymmetric bilinear group assuming its hardness on both groups.
\begin{ass}[XDLIN]
 Given the tuple $\left[ g, g^a, g^b, g^{ac}, g^{bd}, h, h^a, h^b, h^{ac}, h^{bd} \right]$, for random $a,b,c,d\in\Z_q$, the \textit{external DLIN}  assumption in $(q,\G_1,\G_2,\G_T,e)$ states that it is difficult to distinguish  $g^{c+d}$ from a random element of $\G_1$ as well as it is difficult to distinguish $h^{c+d}$ from a random element of $\G_2$.
\end{ass}

\subsubsection{\bf Multilinear Maps}
\begin{defn}
A \textit{$k$-linear map} is a tuple $(q,\G_1,\ldots,\G_k,\{e_{ij} : i,j\ge1, i+j\le k\})$ such that:
\begin{itemize}
\item $e_{ij}:\G_i\times\G_j\rightarrow\G_{i+j}$,
\item each $(q,\G_i,\G_j,\G_{i+j},e_{ij})$ is a bilinear group.
\end{itemize}
Let  $g_i$ be a generator of group $\G_i$ and let $g=g_1$.
\end{defn}

The complexity of many protocols can be reduced if cryptographic multilinear maps are used; see, for instance, \cite{BS03}.
In practice, multilinear maps are challenging to construct,  as mentioned in \cite{GVW15}. Some candidates are presented in \cite{garg2013candidate,coron2013practical,gentry2015graph}, but these multilinear maps suffer from a class of attacks known as “zeroizing” attacks \cite{cheon2015cryptanalysis}, which render them unusable for many applications. To our knowledge, there is only one currently published multilinear map that has not been attacked \cite{ma2018mmap}.

\begin{ass}[MDHE]
Consider $c_1,\ldots,c_k\in\Z_q$ chosen at random and let $\ell \geq 2$ be a natural number. 
The \textit{$(k,\ell)$-Multilinear Diffie-Hellman Exponent} problem for $(q,\G_1,\ldots,\G_k,\{e_{ij}\})$, is, given the following tuple
$$\left[ g^{c_1},\ldots,g^{c_1^\ell},g^{c_1^{\ell+2}},\ldots,g^{c_1^{2\ell}},g^{c_2},\ldots,g^{c_k}\right],$$
to distinguish between $g_k^{c_1^{\ell+1}\prod_{2\le i\le k}c_i}$ and  a random element of $\G_k$.
The MDHE assumption holds if the MDHE problem is hard to solve.
\end{ass}
 Note that for $k=2$ this corresponds to the BDHE problem.

\subsubsection{\bf Bilinear Groups of Composite Order}
We now consider cyclic groups of composite order $n$.
In particular, we consider now  {\em symmetric} bilinear groups $(n,\G_1,\G_T,e)$ as in Definition \ref{def:BilGrp} but with order $n=pq$, where $p$ and $q$ are two distinct primes. In this case, the groups are still cyclic but not every element is a generator. 
If one knows the factorisation of $n$, namely $p$ and $q$, then the following problems lose interest. 
For this reason, the adversary is supposed not to get the primes $p$ and $q$, and we have to assume the hardness of finding a non-trivial factorisation of $n$. 

As noted in \cite{KSW08}, using composite order groups typically gives a simpler scheme; however, since the group sizes are larger, group operations are less efficient.

 Recall that for any finite group $\G$ of order $n$ and for any divisor $m$ of  $n$, there exists a unique subgroup of $\G$ of order $m$. 
 
 \begin{notation}
 In the following, $\G_p$ and $\G_q$ denote the subgroups of $\G_1$ of order $p$ and $q$, respectively. Moreover, $g_p$ and $g_q$ denote random generators of $\G_p$ and $\G_q$, respectively. 
 \end{notation}

In this context, the BDH assumption takes the following form.
\begin{ass}[DcBDH]
The \textit{Decisional composite Bilinear Diffie-Hellman}  assumption for $(n,\G_1,\G_T,e)$  states that, given $\left[ g_q, g_p, g_p^a, g_p^b, g_p^c\right]$, for $a,b,c\in\Z_p$ all chosen at random, it is difficult to distinguish between $e(g_p,g_p)^{abc}$ and a random element in $\G_T$.
\end{ass}

\begin{ass}[Dc3DH]
The \textit{Decisional composite 3-party Diffie-Hellman}  assumption for $(n,\G_1,\G_T,e)$  states that, given the tuple $\left[ g_q, g_p, g_p^a, g_p^b, g_p^{ab} R_1, g_p^{abc} R_2\right]$, for $a,b,c\in\Z_p$, and $R_1,R_2,R_3\in\G_q$, all chosen at random, it is hard to distinguish between $g_p^c R_3$ and a random element in $\G_1$.
\end{ass}

Notice that by knowing $p$ and $q$ and an admissible bilinear map $e:\G_1\times\G_1\rightarrow\G_T$, the assumption is,
as we show now, false. 
Recall that any element $R\in\G_q$ satisfies $R^q=1_{\G_1}$, where $1_{\G_1}$ stands for the identity element of $\G_1$.
Hence, given $\left[ g_q, g_p, g_p^a, g_p^b, g_p^{ab} R_1, g_p^{abc} R_2\right]$ and another element $Z\in\G_1$, we only need to verify whether $e((g_p^{ab} R_1)^q,Z^q)$ equals $e((g_p^{abc} R_2)^q,g_p^q)$.
Indeed if $Z=g_p^c R_3$ then $e((g_p^{ab} R_1)^q,(g_p^c R_3)^q)=e(g_p^{abq},g_p^{cq})=e(g_p,g_p)^{abcq^2}$ and $e((g_p^{abc} R_2)^q,g_p^q)=e(g_p,g_p)^{abcq^2}$.\\

We further consider groups of order $n$, where $n$ is a product of three distinct primes, namely $n=pqr$, and a symmetric bilinear group $(n,\G_1,\G_T,e)$.

\begin{notation}
In the following, $\G_p,\G_q,$ and $\G_r$ denote the subgroups of $\G_1$ of order $p,q$ and $r$, respectively. Moreover, $g_p,g_q,$ and $g_r$ denote, respectively, their generators.
\end{notation}
The following assumption has been introduced in \cite{lewko2010new}.
\begin{ass}[SD3]
The \textit{Subgroup decision problem for three primes} for the bilinear group $(n,\G_1,\G_T,e)$ asks, given the tuple $\left[ g_p,g_r\right]$, to decide whether an element of $\G_1$ is an element of the form $g_p^{a}g_q^{b} \in  \G_{pq}$ or of the form $g_p^{c} \in \G_p$, where $a,c \in \Z_p$, $b \in \Z_q$, and $\G_{pq}$ is the subgroup of order $pq$ in $\G_1$.
In other words, the SD3 problem asks to distinguish whether a given element of $\G_{pq}$ is an element of $\G_p$ or not.
The SD3 assumption holds if the SD3 problem is hard to solve. 
\end{ass}

The following two assumptions are introduced in \cite{KSW08}.
\begin{ass}[D3c3DH] For $(n,\G_1,\G_T,e)$, 
choose random $Q_1,Q_2,Q_3\in\G_q$, random $R_1,R_2,R_3\in\G_r$ and random $a,b,s\in\Z_p$.
Given $$\left[ g_p, g_r, g_qR_1, g_p^b, g_p^{b^2}, g_p^ag_q, g_p^{ab}Q_1, g_p^s, g_p^{bs}Q_2R_2\right],$$
it is supposed to be hard to distinguish between $g_p^{b^2s}R_3$ and $g_p^{b^2s}Q_3R_3$.
We name this assumption the \textit{Decisional 3-composite 3-party Diffie-Hellman}  assumption.
\end{ass}

As for the SD3 assumption, knowing $q$ and $pr$, a non-trivial factorisation of $n$, makes the hypothesis trivial. Indeed, the problem of distinguishing whether an element $T\in\G_1$ is of the form $g_p^{b^2s}R_3$ or $g_p^{b^2s}Q_3R_3$ can be solved by computing $T^{pr}$.
Indeed, $(g_p^{b^2s}R_3)^{pr}=1_{\G_1}$ and $(g_p^{b^2s}Q_3R_3)^{pr}=Q_3^{pr}$.
Moreover, note that if the element $g^{b^2s}$ is somehow known, then the D3c3DH problem can be trivially reduced to the SD3 problem.

\begin{ass}[D3c3BDH]
The \textit{Decisional 3-composite 3-party Bilinear Diffie-Hellman}  assumption for $(n,\G_1,\G_T,e)$ states that, chosen at random $h\in\G_p$, $Q_1,Q_2\in\G_q$ and $s,t\in\Z_q$, given $$\left[ g_p, g_q, g_r, h, g_p^s, h^sQ_1, g_p^tQ_2, e(g_p,h)^t\right],$$ it is hard to distinguish between $e(g_p,h)^{st}$ and a random element in $\G_T$.
\end{ass}

\subsection{Matrix Diffie-Hellman Assumptions}
Let $q \in \N$ be a prime and $\G$ be a multiplicative group of order $q$. Denote by $g$ a generator of $\G$. For an element $a \in \Z_q$, we define $[a] = g^a$ as the implicit representation of $a$ in $\G$. More generally, for a matrix $\mathbf{A} =(a_{ij}) \in \Z_q^{n\times m}$, we define $[\mathbf{A}]$ as the implicit representation of $\mathbf{A}$ in $\G$, that is $[\mathbf{A}] = (g^{a_{ij}})$.
\begin{defn}
Let $\ell, k \in \N$, with $\ell >k$. We call $\mathcal{D}_{\ell,k}$ a \textit{matrix distribution} if its values are 
matrices in $\Z_q^{\ell \times k}$ of full rank $k$. 
\end{defn}
We assume that matrices can be obtain from  $\mathcal{D}_{\ell,k}$  in polynomial time, with overwhelming probability.

\begin{ass}[MDDH]
The \textit{$\mathcal{D}_{\ell,k}$-Matrix  Decisional Diffie-Hellman} assumption for $\G$ states that, given $\left[\mathbf{A}\right]$, where $\mathbf{A}$ is drawn from a matrix distribution $\mathcal{D}_{\ell,k}$, it is hard to distinguish between $[\mathbf{A}\mathbf{v}]$ and $[\mathbf{u}]$, with $\mathbf{v} \in (\Z_q)^k$ and $ \mathbf{u} \in (\Z_q)^{\ell}$ chosen at random.
\end{ass}

The MDDH assumption was introduced in \cite{escala2017algebraic}. The following is its bilateral version, defined for asymmetric bilinear groups $(q,\G_1, \G_2, \G_T,e)$. Let $\mathbf{A} =(a_{ij}) \in \Z_q^{n\times m}$ be a matrix. Then, $[\mathbf{A}]_i$ denotes the implicit representation of $\mathbf{A}$ in $\G_i$, for $i=1,2$.

\begin{ass}[BMDDH]
The \textit{Bilateral $\mathcal{D}_{\ell,k}$-Matrix Decisional Diffie-Hellman}  assumption for $(q,\G_1, \G_2, \G_T,e)$ states that, given $\left[\left[\mathbf{A}\right]_1,\left[\mathbf{A}\right]_2\right]$, where $\mathbf{A}$ is drawn from a matrix distribution $\mathcal{D}_{\ell,k}$, it is hard to distinguish between the two tuples $\left[[\mathbf{A}\mathbf{v}]_1, [\mathbf{A}\mathbf{v}]_2\right]$ and $\left[[\mathbf{u}]_1,[\mathbf{u}]_2\right]$, with $\mathbf{v} \in (\Z_q)^k$ and $ \mathbf{u} \in (\Z_q)^{\ell}$ chosen at random.\\
\end{ass}

\subsection{Dual-Pairing Vector Space (DPVS)}  The DPVS \cite{okamoto2008homomorphic,OT09} is defined by the tuple $( q, \mathbb{V}, \mathbb{W}, \G_T, \mathbb{A}, \mathbb{B}, \tilde{e})$, which is directly constructed from the bilinear group $(q, \G_1, \G_2, \G_T, e)$, where $q$ is a prime. Assume $\G_1, \G_2$ and $\G_T$ are multiplicative cyclic groups of order $q$, and let $g,h$ be generators of $\G_1,\G_2$, respectively. In detail, $\mathbb{V} = (\G_1)^n$ and $\mathbb{W} = (\G_2)^n$ are $n$-dimensional vector spaces over $\mathbb{Z}_q$, where $\textbf{x}=(g^{x_1},\ldots,g^{x_n})\in \mathbb{V}$ and $\textbf{y}=(h^{y_1},\ldots,h^{y_n})\in \mathbb{W}$, with $x_i,y_i\in\mathbb{Z}_q$. Given $\textbf{x}, \textbf{z} \in \mathbb{V}$, the sum $\textbf{x} + \textbf{z}$, required in the definition of a vector space, is here replaced by a component-wise multiplication, that is, $\textbf{x} + \textbf{z} = \left(g^{(x_1+z_1)}, \ldots, g^{(x_n+z_n)}\right)$, and, given $a \in \mathbb{Z}_q$, the scalar multiplication is defined as $a\textbf{x} = (g^{ax_1}, \ldots, g^{ax_n})$. The operations on $\mathbb{W}$ are defined in a similar way. $\mathbb{A}$ and $\mathbb{B}$ are canonical bases for $\mathbb{V}$ and $\mathbb{W}$, respectively.
That is, $\mathbb{A}=\{\textbf{e}_1,\ldots,\textbf{e}_n\}$ where $\textbf{e}_1=(g,1_{\G_1},\ldots,1_{\G_1})$, $\textbf{e}_2=(1_{\G_1},g,1_{\G_1},\ldots,1_{\G_1})$, \ldots, $\textbf{e}_n=(1_{\G_1},\ldots,1_{\G_1},g)$, and $\mathbb{B}$ is defined similarly.
Finally, $\tilde{e}: \mathbb{V} \times \mathbb{W} \rightarrow \G_T$ is a pairing defined by $\tilde{e} (\mathbf{x}, \mathbf{y}) = \prod_{i=1}^n e(g^{x_i}, h^{y_i}) \in \G_T$.\\

Three typical constructions are given in \cite{okamoto2008homomorphic} and \cite{takashima2008efficiently}. Related to DPVS, we find the \textit{Decisional Subspace Problem} (DSP).

\begin{ass}[DSP]
Let $n_1,n_2\in\mathbb{N}$, with  $n_2+1<n_1$.
The \textit{Decisional Subspace Problem} ${\rm DSP}_{(n_1,n_2)}$ assumption for $( q, \mathbb{V}, \mathbb{W}, \G_T, \mathbb{A}, \mathbb{B}, \tilde{e})$ states that, given a base $\A_1=\{\textbf{a}_1,\ldots,\textbf{a}_n\}$ of $\V$,  it is hard to distinguish $\textbf{v}=v_{n_2+1}\textbf{a}_{n_2+1}+\cdots+v_{n_1}\textbf{a}_{n_1}$ from $\textbf{u}=v_1\textbf{a}_1+\cdots+v_{n_1}\textbf{a}_{n_1}$, where   $v_1,\ldots,v_{n_1}$ are chosen at random in $\Z_q$.
\end{ass}

In \cite{okamoto2008homomorphic}, two variants of the DSP assumption are introduced: \textit{Decisional Subspace Problem with Relevant Dual Vector Tuples} (RDSP) and \textit{Decisional Subspace Problem with Irrelevant Dual Vector Tuples} (IDSP). Since a complete description of these two variants is rather long and quite technical, we refer the interested reader to \cite[Chapter 3]{OT09}. Moreover, the authors show that DSP is intractable if the generalised DDH or DLIN problem is intractable.\\

\subsection{Learning With Errors (LWE)}

The LWE problem was introduced by Regev \cite{Reg09}, who showed that solving it on average is as challenging as solving, with a quantum algorithm, several standard lattice problems in the worst case.

For positive $\ell, n, m \in \mathbb{N}$, with $\ell \ge2$, and an error distribution $\chi$ over $\Z_{\ell}$, the (\textit{decisional}) \textit{Learning With Errors} problem is to distinguish between the following pairs of distributions: \begin{center}
$\{A,As+x\}$ and $\{A,u\}$
\end{center}
where $A\in\Z_{\ell}^{n\times m}$, $s\in\Z_{\ell}^n$, $x\in\chi^m$ and $u\in\Z_{\ell}^m$ are all chosen at random.\\

\subsection{Multivariate Quadratic polynomial (MQ) problem} 
The \textit{Multivariate Quadratic polynomial}  problem consists in solving systems of multivariate quadratic equations over finite fields. In detail, denote by $\mathbb{F}_r$ the finite field with $r$ elements, where $r$ is a prime power. Given a system of multivariate polynomials $\{p_1(x_1, \dots, x_n), \dots, p_k(x_1, \dots, x_n)\}$ of degree 2 over $\mathbb{F}_r$, the (computational) MQ problem requires to find  a solution ${\bf x} \in \left(\mathbb{F}_r\right)^n$ such that $p_1({\bf x}) = \cdots = p_k({\bf x}) = 0$. The  MQ problem is proved to be NP-complete \cite{garey1979computers, patarin1997trapdoor}, even for polynomials of degree 2 over $\mathbb{F}_2$.  The relevance, in cryptography, of the MQ problem, relies on its quantum resistance. 

\subsection{Number theory assumptions}

We  mention also two assumptions belonging to the field of number theory.
\begin{defn}
Given three integers $a$, $k$ and $N$,
 $a$ is a \textit{$k$-th residue} modulo $N$ if there exists an integer $b$ such that  $a\equiv b^k \bmod N$.
 If $k$ is equal to 2, then $a$ is called \textit{quadratic residue} modulo $N$. 
\end{defn}

\begin{ass}[QR]
  The \textit{Quadratic Residuosity} problem is to decide, given two integers $a$ and $N$, whether $a$ is a quadratic residue modulo $N$ or not.
  The QR assumption holds if the QR problem is intractable.
\end{ass}
  
Notice that, for an odd prime $p$, it holds
$$
a^\frac{p-1}{2}\bmod p=
\begin{cases}
0 & \mbox{if }a\equiv0 \bmod p,\\
1 & \mbox{if $a$ is a quadratic residue},\\
-1 & \mbox{otherwise}.
\end{cases}
$$

Moreover, for  a prime factor $p$ of $N$, $a$ is a quadratic residue modulo $N$ only if $a$ is a quadratic residue modulo $p$.
Therefore, in the formulation of the QR problem, we assume that no non-trivial factorisation of $N$ is known.

\begin{ass}[DCR]
Let $p,q,p',q'$ be distinct odd primes with $p=2p'+1$ and $q=2q'+1$, and such that $p'$ and $q'$ have the same bit length. Let $N=pq$.
The \textit{Paillier’s Decisional Composite Residuosity} (DCR) assumption states that, given only $N$, it is hard to decide whether an element of $\mathbb{Z}_{N^2}$ is an $N$-th residue modulo $N^2$ or not.
\end{ass}

\section{Security}\label{chap: security}
The idea of security for an FE scheme is that even collusion of secret keys (functional decryption keys) for different functionalities does not reveal anything more about a plaintext than what each secret key allows a user to learn.
Two different models of security for an FE scheme have been used mainly: \textit{indistinguishability-based security}  and \textit{simulation-based security}, shortly IND-based and SIM-based security. 
The former  requires that an adversary cannot distinguish between ciphertexts of any two messages $m_0,m_1$ of its choosing, with access to some secret keys $\mathrm{sk}_f$ of functions $f$ such that $f(m_0) = f(m_1)$. By contrast, the latter requires that the adversary's view can be simulated by a simulator, given only access to secret keys and functions evaluated on the corresponding messages.
 SIM-based security has higher security strength than IND-based security. 
Indeed, SIM-based security implies IND-based security, and there exist  IND-based secure FE schemes for certain functionalities which are not able to be proved secure under SIM-based security, see the example in Subsection \ref{sec: sim-sec}.

In the following, we present in more detail the aforementioned security models and other related notions.

\subsection{Indistinguishability-based security}
In this security game, an adversary $\cA$ receives the public key of the encryption scheme, and then it can request functional decryption keys for functions $f$ of its choice.
$\cA$ also sends two messages, $m_0$ and $m_1$, to the challenger, which samples a random bit $b$ and sends back encryption of the message $m_b$.
We require that the functional decryption keys, obtained by the adversary, are associated with functions $f$ that do not distinguish the two messages, for which $f(m_0)=f(m_1)$.
The system is IND-based secure if the probability that the adversary guesses which bit $b$ was used is close\footnote{Two probabilities are close if their difference is a negligible function concerning the security parameter.} to the probability of a random guess, corresponding to $1/2$.

When the adversary can choose the challenge messages adaptively, we refer to \textit{adaptive security} (also called \textit{full security}). 
In this setting, the game for an adversary $\cA$ can be described as follows.
\begin{itemize}
\item \texttt{Setup}: run $\texttt{Setup}(1^\lambda)\rightarrow({\rm pp},{\rm mk})$ and give $\rm pp$ to $\cA$.
\item \texttt{Query1}: $\cA$ adaptively submits a polynomial number of queries for $f_i \in \cF$, with $i=1,2,\ldots, n$, and it is given $\texttt{KeyGen}(\mathrm{mk},f_i)\rightarrow \mathrm{sk}_{f_i}$.
\item \texttt{Challenge}: $\cA$ submits two messages $m_0,m_1\in X$ satisfying 
\begin{equation}\label{cond1}
f_i(m_0)=f_i(m_1) \mbox{ for all } i=1, \dots, n,
\end{equation}
and it is given $\texttt{Enc}(\mathrm{pp},m_b)$.
\item \texttt{Query2}: $\cA$ continues to issue key queries as in \texttt{Query1} satisfying \eqref{cond1}, and eventually outputs a bit in $\zu$.
\end{itemize}
The adversary wins the game if it guesses correctly the bit $b$.
This notion is sometimes referred  as IND-CPA, where CPA stands for chosen-plaintext attack.
This definition, taken from \cite{BSW11}, is a generalisation of related definitions from \cite{BW07,KSW08}.
An artificial but beneficial weakening of this security model is the so-called \textit{selective security}, where the adversary is required to choose between $m_0$ and $m_1$ beforehand, that is, before obtaining any functional decryption keys or even seeing the public key.
Hence, in this security model the \texttt{Query1} phase is eliminated and the \texttt{Challenge} phase is either right after or right before the \texttt{Setup} phase. 
A
transformation from selective to fully secure (IND-based) functional encryption can be found in \cite{ABSV15} for some cases. 
In literature, we can find modifications of this security game, for example, to a chosen-ciphertext one (CCA) as in \cite{BB04b}.

\subsection{Simulation-based security}\label{sec: sim-sec}
Consider the following generalisation of the FE scheme described in \cite{BSW11}.
For a key space  $K=\{k_1,\ldots,k_t\}$ consider $t$ permutation maps $\sigma_1,\ldots,\sigma_t$ of the plaintext space and the following functionalities $F(k_i,x)=\sigma_i(x)$ for $i=1,\ldots,t$.
Consider now an FE scheme with a trivial encryption algorithm $\texttt{Enc}({\rm pp},x)=x$ and such that $\texttt{Dec}({\rm sk}_i,x)=\sigma_i(x)$, where ${\rm sk}_i$ is the secret key corresponding to key $k_i\in K$.

This scheme clearly leaks more information about the plaintext $x$ than needed. 
However, if we consider an IND-game, the adversary is forced to issue identical challenge messages.
Indeed, for every secret key ${\rm sk}_i$ known by the adversary,  we require $\sigma_i(m_0)=\sigma_i(m_1)$.
This trivially implies that the scheme is IND-secure.
As shown in the example and pointed out by Boneh et al.\ in \cite{BSW11}, IND-based security is vacuous and inadequate for specific functionalities. 
This indicates that we should opt for simulation-based security whenever possible. In the definition of SIM-based security, an efficient simulator is required to generate the view of the adversary in the security game, only knowing the information that leaks from the encrypted values and corrupted functional decryption keys (possessed by the adversary).
This security notion for FE schemes was first given in \cite{ON10,BSW11}.
In the following, we report in more detail the idea of SIM-based security, as given in \cite{BSW12}.

Let $\cA$ be a polynomial time adversary that takes as input the public key $\rm pp$, a set of secret keys $\{ {\rm sk}_i=\texttt{KeyGen}(\mathrm{mk},f_i) : 1\le i\le\ell \}$ for functions $f_i\in\cF$ of its choice, and a ciphertext $c=\texttt{Enc}(\mathrm{pp},x)$.
Then $\cA$ might output some information about the decryption of $c$.
The system is SIM-based secure if, for every such $\cA$ there is another polynomial time algorithm $\cB$, called  \textit{simulator}, that given $\rm pp$ and $f_1(x), \ldots, f_\ell(x)$, but not given $c$, is able to output the same information about $x$ that $\cA$ outputs.
This implies that knowing the ciphertext $c$ does not give any additional information about  $x$, beyond that obtained from  $f_1(x), \ldots, f_\ell(x)$.

We refer the reader to \cite{BSW11} for a more formal definition of SIM-based security. We also refer the reader to \cite{Gay19} for the distinction between adaptive and selective SIM-based security.

In \cite{ON10,BSW11}, it is shown that this security notion is impossible to achieve in many cases.
For instance, \cite[Theorem 2]{BSW11} states that there are no SIM-based secure IBE schemes in the non-programmable random oracle model 
(see Subsection \ref{subsec: random oracle}).
However, the random oracle model is more promising.  In fact, in this model, it is possible to achieve SIM-based security.
 For example, as exhibited in \cite[Theorem 4]{BSW11}, given an IND-based secure predicate encryption system with a public index, it is possible to convert the system into a scheme SIM-based secure in the random oracle model.

\subsection{Function privacy}

Although the vast majority of research has focused on the privacy of the encrypted messages, it is crucial to offer privacy for the functions for which decryption keys are provided in many realistic scenarios. Indeed, the functions themselves may already leak critical or sensitive information. \textit{Function privacy} was first proposed by Shen, Shi, and Waters in the setting
of private-key predicate encryption in \cite{shen2009predicate}. It requires that a secret key $\mathrm{sk}_f$ does not reveal anything about the function $f$.
For public-key functional encryption, standard function privacy cannot be generally achieved.
Indeed,  given a key $\mathrm{sk}_f$, an adversary can generate ciphertexts of  messages $x_1,\dots, x_n$ of its choices, and thus it obtains the values $f(x_1),\dots, f(x_n)$.

Boneh, Raghunathan, and Segev in \cite{BRS13a}, and later in \cite{BRS13b}, address this problem for some public-key predicate encryption schemes.
In \cite{BRS13a}, they construct function-private IBE which implies predicate-private encrypted keyword search. Whereas, in \cite{BRS13b}, they consider the notion of \textit{subspace-membership encryption}, a generalisation of inner-product encryption that supports subspace-membership predicates. They present a generic construction of a function-private subspace-membership encryption scheme based on any inner-product encryption scheme. Further analysis on HVE can be found in \cite{BC+19}.
\subsection{Random oracle model}\label{subsec: random oracle}

The \textit{Random Oracle Model} (ROM) was introduced by Bellare and Rogaway in \cite{BR93} as an alternative model to analyse the security of certain cryptographic constructions. A \textit{random oracle} is a function $\mathcal{O}:X\rightarrow Y$ chosen uniformly at random from the set of all such functions, where  $Y$ is supposed to be a finite set. ROM assumes the availability of a random oracle to all parties (adversary, challenger, etc.) playing a role in any game aimed
at studying the security of a cryptographic primitive. The introduction of ROM made it possible to prove the security of many different kinds of cryptographic primitives. As truly random functions do not exist in practice, when using these primitives in the real world, a secure hash function plays the role of a random oracle. The heuristic is that secure hash functions are close enough to random oracles in their behavior, so the primitives remain secure even under this substitution. 

In ROM, the challenger simulates the random oracle for the adversary in
the security game. The challenger can select answers dynamically for the queries that parties make to the random oracle; 
 we refer to this property as \textit{programmability} of the random oracle. On the other hand, the \textit{non-programmable random oracle model} restricts the challenger such that it can no longer choose return values of the random oracle. This security restriction makes the underlying proofs preferable, as they rely on fewer properties of the random oracle and so can be considered as a step toward getting rid of it. A more detailed description of these models can be found in \cite{FLR+10}.

ROM is in opposition with the \textit{standard model} in which the adversary is only limited by the amount of time and computation available.
A scheme is secure in the standard model if it can be proved secure only using complexity assumptions.

Another model used is the \textit{generic-group model}, which assumes that the properties of the representation of the elements of the group under consideration cannot be exploited. This extends naturally also to bilinear maps.

\section{IBE and ABE}\label{chap: IBEandABE}
Historically, Identity-Based Encryption (IBE) has been the first attempt to construct a functional encryption scheme beyond traditional public-key encryption. IBE schemes have been then generalised into Attribute-Based Encryption (ABE) schemes.
Classical IBE and ABE systems are also referred to as predicate encryption systems with a public index or \textit{payload hiding} systems. Indeed, the payload message is hidden, but the index related to it (an identity, a set of attributes, or an access policy) is often in the clear.
In the following, we define IBE and ABE systems and briefly introduce some of the known constructions.

\subsection{Identity-Based Encryption}\label{IBE}
In an IBE scheme we assume that a plaintext to be encrypted is a pair $({\rm ID},m)\in I\times M$, where $\rm ID$ is the identity and $m$ is the payload message.
Secret keys are also associated with an identity $\rm ID^\prime$, and decryption succeeds to recover the payload message if there is a match between the identity associated with the ciphertext and the identity associated with the secret key, that is,
$$F({\rm ID}^\prime, ({\rm ID},m))=\begin{cases}
m & {\rm if}\ {\rm ID}={\rm ID}^\prime,\\
\perp & {\rm otherwise}.
\end{cases}$$
For example, in an email system, the identity can be the email address and the payload can be the email content.
Employing a public-key cryptosystem, it is possible to encrypt a message to any user whose email address is known.
Using his identity, the receiver can obtain the corresponding secret key and decrypt all the received messages.
The concept of IBE was first thought of by Shamir in 1984 \cite{Sha84} as a way to simplify certificate management in email systems.
The first realisations of IBE schemes were given in 2001: the security of these schemes was based on either the BDH or the QR assumption in the random oracle model \cite{BF01,Coc01}.
In \cite{BF01}, the system is based on bilinear maps between groups (the Weil pairing on elliptic curves is an example of such maps), and it is IND-based secure against an adaptive chosen ciphertext attack.
In \cite{Coc01}, the system is based on quadratic residues; it is relatively easy to implement, but the ciphertext length is very large (see \cite{BGH07} for an improvement).

Since then, tremendous progress has been made towards obtaining IBE schemes that are secure in the standard model based on pairings as well as on lattices \cite{CHK03, BB04, BB04b,Wat05,CS05,Gen06,Nac07,ABB10,Wat09,CW13}.
Different assumptions are used to prove the security of the cited constructions, such as DBDH, DBDHI, truncated ABDHE, DLIN, and LWE.

Several surveys on IBE schemes have been published during the years. We refer the interested reader to \cite{Boy08,HY18,TAA19}.

\subsubsection{\bf Hierarchical IBE}\label{subsec: HIBE}
Some of the mentioned papers also present constructions of HIBE.
With Hierarchical IBE (HIBE), we refer to a generalisation of IBE that mirrors an organisational hierarchy.
Identities are taken as vectors, and there is a fifth algorithm called \texttt{Derive}.
A vector of dimension $\ell$ represents
an identity at level $\ell$ of the hierarchy.
Algorithm \texttt{Derive} takes as input an identity ${\rm ID}=(I_1,\ldots,I_\ell)$ at level $\ell$ and the private key of the parent identity ${\rm ID}^\prime=(I_1,\ldots,I_{\ell-1})$ at level $\ell-1$.
It outputs the private key for an identity $\rm ID$.
So an identity can issue private keys to its descendants' identities, but it cannot decrypt messages intended for other identities.
The notion of HIBE was first defined in \cite{HL02} and constructions in the random oracle model were given in \cite{GS02} (based on the DBDH assumption).
Several works on HIBE schemes followed, and their security is proved in the standard model, see \cite{CHK03,BB04,BBG05,ABB10,Wat09,CHKP10}.
These schemes are based on different security assumptions such as  DBDH, DBDHE, DLIN, and LWE.
For a deeper analysis of HIBE schemes, we refer the reader to \cite{DRS17}.

\subsection{Attribute-Based Encryption}\label{sec: ABE}
In many situations, it is required that encrypted sensitive data can be decrypted only by users who have certain credentials or attributes.
As an example, consider a company with different branches \{Trento, London, Oslo, etc.\} where employees are divided into management-levels $\{1,2,\ldots,5\}$ and characterised by the years spent working for the company.
A sensitive memo about a project in the Italian branch may be encrypted with the following access policies: 
(branch = Trento \texttt{AND} years $\ge3$)
 \texttt{OR} (management-level = $5$), and only employees with these attributes can read the memo.

To comprehend also these kinds of scenarios, ABE systems are introduced as an extension of IBE.
In ABE schemes, sets of attributes or access policies over attributes are used instead of the identities in the encrypted message and the secret key.
ABE integrates encryption and access control, and it is ideal for sharing secrets among groups, especially in a cloud environment.
Many different ABE schemes with additional features have been proposed during the years.

In the following, we describe some of the main features of ABE schemes.
To have a more detailed overview of ABE systems, we refer the reader to \cite{QLDJ14,LCH13,ZD+20}.

 \subsubsection{\bf Threshold-Policy Access Control}
Sets of descriptive attributes ($\omega,\omega^\prime$) label both the ciphertext and the user's secret key.
The decryption is successful if there is a sufficient overlap, say $d$, between the attribute sets of the ciphertext and the secret key, that is,
$$F(\omega^\prime,(\omega,m))=\begin{cases}
m & {\rm if}\ |\omega\cap\omega^\prime|\ge d,\\
\perp & {\rm otherwise}.
\end{cases}$$
This threshold provides error tolerance for coarse-grained access control schemes, and it is particularly useful to biometric applications where public information, not a secret, is treated as the input.
This type of ABE scheme was introduced by Sahai and Waters in 2005 \cite{SW05} called by the authors \textit{Fuzzy Identity-Based Encryption} scheme.
This scheme is based on polynomial interpolation, and it is secure under the decisional MBDH assumption.
Since the work of Sahai and Waters, several other constructions followed; for example, in \cite{Cha07} there is an extension of the Fuzzy-IBE scheme based on the DBDH assumption.

\subsubsection{\bf Key-Policy Access Control}
In Key-Policy ABE (KP-ABE), an access policy $\phi$ is encoded into a user's secret key, and ciphertext is labeled with a set of descriptive attributes $\omega$.
The access structure of the private key specifies which type of ciphertexts the key can decrypt, that is,
$$F(\phi,(\omega,m))=\begin{cases}
m & {\rm if}\ \omega\ {\rm satisfies }\ \phi,\\
\perp & {\rm otherwise}.
\end{cases}$$
In KP-ABE schemes, the data owner has limited control over who can decrypt the data. KP-ABE schemes are suitable for structured organisations with rules about who may read particular documents. Typical applications of KP-ABE include secure forensic analysis and target broadcast.
The first KP-ABE was presented by Goyal, Pandey, Sahai, and Waters in 2006 \cite{GPSW06}.
Further results on KP-ABE followed, see for example \cite{AHL+12,OT10,GVW13,BGG+14}.
KP-ABE schemes are based on different assumptions, such as
DBDH, MDHE, DLIN and LWE.

\subsubsection{\bf Ciphertext-Policy Access Control}
In Ciphertext-Policy ABE (CP-ABE), a user's secret key is associated with a list of attributes $\omega$, and a ciphertext specifies an access policy $\phi$.
A ciphertext can be decrypted by a user if and only if the user's attributes satisfy the ciphertext's access policy, that is,
$$F(\omega,(\phi,m))=\begin{cases}
m & {\rm if}\ \omega\ {\rm satisfies}\ \phi,\\
\perp & {\rm otherwise}.
\end{cases}$$
Hence, the data owner can determine who can decrypt and, if the policy needs to be updated frequently, CP-ABE can be more flexible than KP-ABE.
One of the first constructions of CP-ABE appears in \cite{BSW07}, which is based on the random oracle assumption.
Improvements on CP-ABE schemes can be found in several papers, 
for example, in \cite{Wat11,LOS+10,OT10}.
The cited schemes are based on different assumptions, such as the decisional PBDHE, SD3, DLIN.

\subsubsection{\bf Non-monotonic Access Control}
Usually, KP-ABE and CP-ABE use monotonic access structures, which means that users cannot express the negative attribute to exclude certain users from decryption.
The number of attributes could be doubled to introduce the NOT expression: for each original attribute, the system must assign one corresponding negative attribute. This approach is not efficient in the number of attributes to consider.
The first non-monotonic attribute-based scheme to reduce the number of negative attributes can be found in \cite{OSW07}.
  Another non-monotonic system is presented in \cite{AHL+12}.
 The mentioned schemes are based on the DBDH and the DBDHE assumptions.
  
\subsubsection{\bf Hierarchical ABE}\label{subsec: HABE}
In  Hierarchical ABE (HABE), as for HIBE, the delegation of access privilege is hierarchical.
This feature can be helpful when there are many attributes in the system. It relieves the authority (in charge of generating secret keys) from the heavy key management burden.
Users are associated with attribute vectors that denote their depth in the hierarchy. 
Users at a higher level may delegate secret keys to their subordinates.
Weng et al.\cite{WLW10} combined HIBE with CP-ABE to achieve HABE.
Other HABE schemes can be found in \cite{DW+14,LYZ19}.

\subsubsection{\bf Revocable ABE}
The revocation functionality is realised in revocable ABE schemes, where it is possible to revoke malicious users from the system or revoke an expired attribute.
For more details, see \cite{HN10,LC+13,GAS16,AS+19,CH+20}.

\subsubsection{\bf Other features}
Other important features characterise ABE schemes. 
\textit{Scalability} improves the dynamism of the system:
a new user with new attributes is allowed to join the system without restarting. 
\textit{Key delegation} schemes are characterised by having another algorithm in their definition: $\texttt{Delegate}({\rm sk}_k,\tilde{k})\rightarrow{\rm sk}_{\tilde{k}}$, on input a user's secret key ${\rm sk}_k$ and a set of attributes $\tilde{k}\subseteq k$, outputs a new user's secret key ${\rm sk}_{\tilde{k}}$.  Hence a user can create new secret keys for subsets of its attributes set. Both features can be achieved with a hierarchical structure.
 Another important feature is the \textit{Multi-Authority}. Multi-Authority ABE schemes are described in Subsection \ref{sec: multi-authority}.
 \textit{Accountability} helps to track misbehaving users by reducing trust assumptions on both authorities and users.
 For more details on these features, we refer the interested reader to the mentioned surveys on ABE schemes  \cite{QLDJ14,LCH13,ZD+20}.

\section{Predicate Encryption with private index}\label{chap: PE}

As briefly mentioned in Section \ref{chap: def FE}, in a PE scheme a ciphertext is associated with a set of attributes $\textbf{x}$ (called also \textit{index}),  while a decryption secret key is tied to a specific predicate $P$ in some function class $\mathcal{F}$.
The secret key ${\rm sk}_P$, sometimes called also \textit{token}, correctly decrypts a ciphertext if and only if the associated set of attributes $\textbf{x}$ satisfies $P(\textbf{x})=1$, that is,
$$F(P,(\textbf{x},m))=\begin{cases}
m & {\rm if}\ P(\textbf{x})=1,\\
\perp & {\rm otherwise}.
\end{cases}$$
 Note that both IBE and ABE can be cast in the framework of predication encryption. In particular, IBE schemes can be viewed as predicate encryption for the class of equality tests.
A ciphertext in PE hides the message (or payload) $m$:
an adversary holding keys ${\rm sk}_{P_1},\cdots,{\rm sk}_{P_\ell}$ learns nothing about a message encrypted using attribute $\textbf{x}$ only if $P_1(\textbf{x})=\ldots=P_\ell(\textbf{x})=0$.
This security notion is often referred to as \textit{payload hiding}.
 If the ciphertext hides also the attributes $\textbf{x}$, we talk about \textit{attribute hiding} or \textit{private-index PE}. Whereas, IBE and ABE are also called \textit{public-index PE}.
Some proposed constructions only hide the attributes underlying each ciphertext when the decryption is not successful. This is referred to as \textit{weakly-hiding the attributes}.
Instead, in \textit{fully-hiding} constructions, the only information that leaks is the value of the predicate evaluation, even when the decryption succeeds. In the following, we present the three main classes of private-index PE schemes.

\subsection{Anonymous Identity-Based Encryption}
With anonymous-IBE, we refer to an IBE scheme in which the ciphertext identity (the index related to the message) is hidden, and one can only determine it by using the corresponding private key (attribute hiding).
Most of the IBE schemes presented in Subsection \ref{IBE} are only payload hiding.
The IBE system proposed in \cite{BF01} is inherently anonymous, as noticed by Boyen in \cite{Boy03}.
Recall that its security proofs are set in the random oracle model.
In the full version of  \cite{BSW11}, Boneh, Sahai, and Waters show that this system \cite{BF01} is SIM-based secure (in the random oracle model).
In \cite{BW06}, we can find an identity-based cryptosystem that features fully anonymous ciphertexts and hierarchical key-delegation. The proof of security is given in the standard model, based on the  DLIN assumption in bilinear groups.
The system has small ciphertexts of size linear in-depth of the hierarchy.
Other anonymous IBE constructions, both pairing-based (truncated ABDHE, DBDH, DLIN assumptions) and lattice-based (LWE assumption), can be found in \cite{Gen06, CHKP10,ABB10, Pei09,CK+09}.
More recent studies on anonymous IBE can be found in \cite{HK+20,BBP19,FT15,HW+16,MWL18,XL+16}.

\subsection{Hidden Vector Encryption}\label{subsec:HVE}
Hidden Vector Encryption (HVE) is a particular case of Predicate Encryption (PE).
Boneh and Waters first proposed it in \cite{BW07} as a subclass of Searchable Encryption systems. 
Previous Searchable Encryption schemes were proposed in \cite{BCOP04} supporting simple equality tests.
Let $\Sigma$ be a finite set of attributes, let $*$ be a special character (not in $\Sigma$) and set $\Sigma_*=\Sigma\cup\{*\}$. We refer to $*$ as a \textit{wildcard character}.
Often $\Sigma$ corresponds to the set of binary strings of arbitrary length.
For a vector $\boldsymbol{\sigma}=(\sigma_1,\ldots,\sigma_n)\in\left(\Sigma_*\right)^n$, a predicate $P_{\boldsymbol{\sigma}}$ over $\Sigma^n$ is defined as follows: for $\textbf{x}=(x_1,\ldots,x_n)\in\Sigma^n$
$$P_{\boldsymbol{\sigma}}(\textbf{x})=\begin{cases}
1 & \mbox{if } \forall i\  \sigma_i=x_i \mbox{ or } \sigma_i=*,\\
0 & \mbox{otherwise}.
\end{cases}$$
Based on \cite{BW07}, an HVE scheme supporting equality predicates leads to an anonymous IBE scheme where a ciphertext reveals no useful information about the receiver's identity.

In \cite{BW07}, an HVE scheme is constructed 
based on composite-order bilinear groups.
It supports evaluation of conjunctive equality ($x_i=\sigma_i$), comparison ($x_i\ge\sigma_i$) and subset predicates ($x_i\in A_i\subseteq\Sigma$).
Additionally, it supports a conjunctive combination of these primitive predicates by extending the size of
ciphertexts.
The ciphertext size is $\mathcal{O}(n)$ and the token size is $\mathcal{O}({\rm weight}(\boldsymbol{\sigma}))$, where ${\rm weight}(\boldsymbol{\sigma})={\rm weight}\left(\left(\sigma_1,\ldots,\sigma_n\right)\right)=|\{\sigma_i : \sigma_i\ne*\}|$.
The scheme is proved selectively IND-based secure under the Dc3DH  assumption and the DcBDH assumption, where an adversary has restrictions on the queries it can ask.
After the introduction of HVE based on composite-order bilinear groups, several other HVE schemes
have been proposed, see e.g. \cite{IP08,SW08,Par10,PL+13,DIP12,ZYT13,SV+10}.
 The mentioned schemes are based on pairings and rely on different assumptions, such as DLIN and DBDH.

\subsection{Inner-Product Predicate Encryption}\label{sec: IP-PE}
Researchers have tried to expand the function class $\mathcal{P}$ supported by a PE scheme, with the ultimate goal being to handle all polynomial-time predicates.
A stepping stone towards this goal is to support the predicated corresponding to inner products over $\Z_N$ (for some large integer $N$).
For a set of attributes $\Sigma=(\Z_N)^n$, we consider the class of predicates $\mathcal{P}=\{P_\textbf{x}\ |\ \textbf{x}\in(\Z_N)^n\}$ with 
$$P_\textbf{x}(\textbf{y})=\begin{cases}
1 & \mbox{if } \langle\textbf{x},\textbf{y}\rangle=\sum_{i=1}^nx_i\cdot y_i \equiv 0\bmod N,\\
0 & \mbox{otherwise}.
\end{cases}$$
Notice that sometimes in literature, these schemes are called IPE (Inner Product Encryption) schemes. This can create confusion with schemes that output the value of the inner product (see Section \ref{chap: IPE}).
We then refer to a scheme supporting this class of functionalities as an Inner-Product Predicate Encryption (IP-PE) scheme.
In \cite{KSW08}, Katz, Sahai, and Waters construct an IP-PE scheme without a random oracle, based on two assumptions in composite-order groups equipped with a bilinear map.
The (decisional) assumptions D3c3DH and D3c3BDH, introduced in \cite{KSW08}, are shown to hold in the generic-group model, assuming that finding a non-trivial factorisation of the group order is hard.
The scheme is proved selective IND-based secure. 
In the same work, the authors show that any PE scheme supporting inner product predicates can be used as a building block to construct predicates of more general types.
IP-PE implies anonymous-IBE and HVE, and
it can construct a PE scheme supporting polynomial evaluation (polynomials of some bounded degree).
This allows us to build PE schemes for disjunctions and conjunctions of equality tests and Boolean variables, hence arbitrary CNF and DNF formulas (bounded depth).
IP-PE can also be used to construct an (attribute hiding) fuzzy-IBE with an exact threshold policy.
We report here some of the mentioned implications, as presented in \cite{KSW08}.
Anonymous IBE can be recovered from IP-PE in the following way. 
To generate secret keys for identity $I\in\Z_N$, set $\textbf{x}=(1,I)$ and output the secret key for the predicate $P_\textbf{x}$. To encrypt a message $M$ for identity $J\in\Z_N$, set $\textbf{y}=(-J,1)$ and encrypt the message with index $\textbf{y}$. Then decryption is successful if $\langle \textbf{x},\textbf{y}\rangle=0$, that is, if $I=J$.
HVE for predicates over $\Sigma^\ell=(\Z_N)^\ell$ can be recovered from IP-PE by doubling the dimension.
Consider the predicate $P_{\boldsymbol{\sigma}}$ described in Subsection \ref{subsec:HVE}. To generate a secret key corresponding to $P_{\boldsymbol{\sigma}}$, first construct a vector $\textbf{x}=(x_1,\ldots,x_{2\ell})$ as follows:
\begin{align*}
    {\rm when}\ \sigma_i\ne*, &\text{ set } x_{2i-1}=1,\ x_{2i}=\sigma_i,\\
    {\rm when}\ \sigma_i=*, &\text{ set } x_{2i-1}=0,\ x_{2i}=0.
\end{align*}
Then output the secret key for predicate $P_\textbf{x}$.
To encrypt a message $M$ for attribute $\textbf{z}=(z_1,\ldots,z_\ell)$, choose random $r_1,\ldots,r_\ell\in\Z_N$, construct a vector $\textbf{y}_\textbf{r}=(y_1,\ldots,y_{2\ell})$  with
$$y_{2i-1}=-r_i\cdot z_i,\ y_{2i}=r_i,$$
and output the corresponding ciphertext.
Clearly, if $P_{\boldsymbol{\sigma}}(\textbf{z})=1$ then $\langle\textbf{x},\textbf{y}_\textbf{r}\rangle=0$ and $P_\textbf{x}(\textbf{y}_\textbf{r})=1$.
On the other side, assuming $\gcd(\sigma_i-z_i,N)=1$, if $P_{\boldsymbol{\sigma}}(\textbf{z})=0$ then the probability that $\langle\textbf{x},\textbf{y}_\textbf{r}\rangle=0$ is $1/N$.
Finally, polynomial evaluation can be described with IP-PE.
Let $a\in\Z_N[x]$ be a polynomial of degree at most $d$, $a(x)=a_dx^d+\ldots+a_1x+a_0$, and consider the predicate 
$$P_a(x)=\begin{cases}
1 & {\rm if}\ a(x)=0,\\
0 & {\rm otherwise}.
\end{cases}$$
To generate a secret key corresponding to $a$, set $\textbf{x}=(a_d,\ldots,a_0)$ and output the key for predicate $P_\textbf{x}$.
To encrypt a message $M$ for the attribute $w\in\Z_N$, set $\textbf{y}=(w^d\mod N,\ldots,w^0\mod N)$ and output the corresponding ciphertext.
The decryption is then successful if $\langle\textbf{x},\textbf{y}\rangle=0$, that is, if $a(w)=0$.

The first lattice-based construction of a PE scheme for the inner product is proposed in \cite{AFV11}.
The scheme is proved secure (weakly-hiding the attributes) under the LWE assumption.
Other constructions of PE schemes supporting inner products can be found in \cite{OT09,LOS+10,Par11,OT12,Wee17}.
The above schemes are based on pairings and rely on different assumptions such as RDSP, IDSP, DLIN, and DBDH.

\section{Beyond predicate encryption: Inner Product Encryption}\label{chap: IPE}

So far, we have only discussed particular kinds of functional encryption where decryption successfully recovers the entire message if the attributes associated with the ciphertext (resp.\ the functional decryption key) satisfy the access policy embedded in the key (resp.\ the ciphertext). While this is a fruitful generalisation of traditional public-key encryption, since it permits embedding complex access policies into the encrypted data, it is still an ``all-or-nothing'' encryption. 
The \textit{Inner Product Encryption} (IPE), introduced in \cite{ABDP15}, is the first FE scheme with fine-grained access to the encrypted data, where decryption recovers partial information about the encrypted data. 

In an IPE scheme, a ciphertext $c_{\mathbf{x}}$ is related to a vector $\mathbf{x} =(x_1, \dots, x_n) \in (\mathbb{Z}_q)^n$, and a functional decryption key $s_{\mathbf{y}}$ is related to a vector $\mathbf{y} =(y_1, \dots, y_n)\in (\mathbb{Z}_q)^n$. Given a ciphertext and a secret key, the decryption algorithm computes the inner product $\left\langle \mathbf{y}, \mathbf{x} \right\rangle = \sum_{i=1}^n y_i x_i \in \mathbb{Z}_q$, that is,
\[
F({\bf y},{\bf x})= \sum_{i=1}^n y_i x_i \, \bmod q.
\]
Notice that  IPE is different from the inner-product predicate encryption presented in Subsection \ref{sec: IP-PE}, where the ciphertext of a message $m$ is related to a vector $\mathbf{x} \in \left(\mathbb{Z}_q\right)^n$, the secret key is related to a vector $\mathbf{y} \in \left(\mathbb{Z}_q\right)^n$, and the decryption algorithm reveals $m$ if and only if $\left\langle \mathbf{x}, \mathbf{y} \right\rangle \equiv 0 \bmod q$. By contrast, the output in IPE schemes is the actual value of the inner product. At first glance, this functionality might seem less powerful than the previous definitions and constructions of identity-based encryption and attribute-based encryption. Indeed, inner product encryption cannot be used to construct those kinds of cryptosystems. However, neither IBE nor ABE can be used to build IPE. So this notion is completely separated from previous works, and its relevance stands in providing only partial information on the plaintext and not ``all-or-nothing''. Moreover, IPE directly finds many natural applications in theory as well as in practice. First, it can be used directly to compute the Hamming distance between two words, and more generally, the inner product is a distance for vectors of the fixed norm. So this could be exploited, for instance, in biometric identification. IPE could also be used in machine learning, where, for instance, Support Vector Machine evaluation can be just an inner product between a test vector and the data. In this case, a secret key could classify the data concerning one criterion and learn nothing about the rest of it. It also allows making statistical analysis, like weighted means or by encrypting the cross-products between the entries. It can even be used for computations of variance and other quadratic polynomials.

In \cite{ABDP15}, Abdalla et al.\ present a direct construction of public-key IPE relying on standard assumptions, such as LWE and DDH. The construction is only proved to be secure against selective adversaries who are asked to commit to their challenges at the beginning of the security game. In \cite{agrawal2016fully}, Agrawal et al.\ provide constructions that provably achieve security against more realistic adaptive attacks, where the messages $m_0$ and $m_1$ may be adaptively chosen in the challenge phase. Later, in \cite{agrawal2020adaptive}, Agrawal et al.\ prove that the scheme in \cite{agrawal2016fully} achieves adaptive SIM-based security (AD-SIM) rather than just adaptive IND-based security. Moreover, some IPE schemes based on the DDH, DCR, and LWE assumptions prove AD-SIM security for an unbounded number of key queries and a single challenge ciphertext.

In \cite{benhamouda2017cca}, Benhamouda et al. propose a generic construction for the first IND-CCA IPE schemes based on the DDH, DCR, and any of the MDDH assumptions. In \cite{castagnos2018practical}, Castagnos et al. provide the first adaptive IND-secure IPE schemes, which allow for the evaluation of unbounded inner products modulo a prime $p$. The constructions rely on the assumption that supposes a DDH group containing a subgroup where the discrete logarithm problem is easy.

In \cite{chotard2020dynamic}, Chotard et al.\ present an IPE scheme for computing some weighted sums on aggregated encrypted data from standard assumptions in prime-order groups in the random oracle model.

Security of the mentioned IPE schemes is based on traditional number-theoretic assumptions, such as the hardness of factorisation and the discrete logarithm problem. However, the intractability of these number-theoretic problems will be extremely vulnerable due to Shor's algorithm \cite{shor1999polynomial} if the powerful quantum computers are built. In a recent work \cite{debnath2020post}, Debnath et al.\ propose a multivariate cryptography-based inner-product encryption scheme. It is the first of its kind in the context of multivariate public-key cryptosystems. The authors use the technique of identity-based signature of \cite{chen2019identity} to develop this interesting scheme. The scheme attains non-adaptive SIM-based security under the hardness of the MQ problem.

\subsection{Function-hiding inner product encryption} As reported in \cite{bishop2015function}, in many real scenarios it is important to  consider also the privacy of the computed function. Consider the following motivating example: suppose a hospital subscribes to a cloud service provider to store its patients' medical records. To protect the privacy of the data, these records are stored in an encrypted form. At a later point in time, the hospital can request the cloud to perform some analysis on the encrypted records by releasing a decryption key $\mathrm{sk}_f$ for a function $f$ of its choice. If the FE scheme in use does not guarantee any hiding of the function (which is the case for many existing FE schemes), then the key $\mathrm{sk}_f$  might reveal $f$ completely to the cloud, which is undesirable when $f$ itself contains sensitive information. This has motivated the study of function privacy in FE, see for instance \cite{shen2009predicate, BRS13a, brakerski2018function}. 

An IPE scheme is called \textit{function-hiding} if the keys and ciphertexts reveal no additional information about the related vectors beyond their inner product. The \textit{fully function-hiding} IPE achieves the most robust IND-based notion of both data and function privacy in the private-key setting in the standard model. 
We sketch here the model of full function privacy, as described in \cite{agrawal2015practical, kim2019new}.
Adversaries are allowed to interact with two \textit{left-or-right} oracles $\mathrm{KeyGen}_b(\mathrm{mk}, \cdot, \cdot)$ and $\mathrm{Enc}_b(\mathrm{mk}, \cdot , \cdot)$ for a randomly chosen $b \in \{0,1\}$, where $\mathrm{KeyGen}_b$ takes two functions $f_0$ and $f_1$ as input and it returns a functional decryption key $\mathrm{sk}_{f_b}=\texttt{KeyGen}({\rm mk},f_b)$. The algorithm $\mathrm{Enc}_b$ takes two messages $x_0$ and $x_1$ as input and it outputs a ciphertext $\mathrm{c}_{x_b} =\texttt{Enc}({\rm mk},x_b)$. Adversaries can adaptively interact with oracles for any polynomial (a priori unbounded) number of queries. To exclude inherently inevitable attacks, there is a condition for adversarial queries that all pairs $(x_0, x_1)$ and $(f_0, f_1)$ must satisfy $f_0(x_0) = f_1(x_1)$.

Only two approaches have been proposed for (fully) function-private IPE schemes in the private-key setting. One is to employ the Brakerski-Segev general transformation from (non-function-private) FE schemes for general circuits \cite{brakerski2018function}. The transformation itself is efficient since it simply combines symmetric key encryption with FE in a natural manner. Anyway, this approach requires computationally intensive cryptography tools, such as IND obfuscation\footnote{In simple words, obfuscation aims to make a computer program ``unintelligible'' while preserving its functionality, and indistinguishability obfuscation requires that, given any two equivalent circuits of similar size, their obfuscation should be computationally indistinguishable.}, to realise non-function-private FE for general circuits, meaning it may be relatively inefficient overall. The other approach may be more practical. It directly constructs IPE schemes by using the dual-pairing vector spaces (DPVS) introduced by Okamoto and Takashima \cite{okamoto2008homomorphic,OT09}. 

In the last few years, there has been a flurry of works on the construction of function-hiding IPE, starting with the work of Bishop, Jain, and Kowalczyk \cite{bishop2015function}. They propose a function-hiding IPE scheme under the SXDH assumption, which satisfies an adaptive IND-based security definition. But the security model has one limitation: all ciphertext queries $\mathbf{x}_0, \mathbf{x}_1$ and all secret key queries $\mathbf{y}_0, \mathbf{y}_1$ are restrained by $\left\langle \mathbf{x}_0, \mathbf{y}_0\right\rangle = \left\langle \mathbf{x}_0, \mathbf{y}_1\right\rangle = \left\langle \mathbf{x}_1, \mathbf{y}_0\right\rangle =\left\langle \mathbf{x}_1, \mathbf{y}_1\right\rangle$. In \cite{datta2017strongly}, Datta et al.\ develop a full function-hiding IPE scheme built in the setting of asymmetric bilinear pairing groups of prime order. The security of the scheme is based on the well-studied SXDH assumption where the restriction on adversaries' queries is only $\left\langle \mathbf{x}_0, \mathbf{y}_0\right\rangle = \left\langle \mathbf{x}_1, \mathbf{y}_1\right\rangle$. Here, secret keys and ciphertexts of $n$-dimensional vectors consist of $4n +8$ group elements. Tomida et al., in \cite{tomida2016efficient}, construct a more efficient function-hiding IPE scheme than that of \cite{datta2017strongly} under the XDLIN assumption, where secret keys and ciphertexts consist of $2n +5$ group elements. Kim et al., in \cite{kim2018function}, put forth a fully-secure function-hiding IPE scheme with smaller parameter sizes and run-time complexity than those in \cite{bishop2015function, datta2017strongly}. The scheme is proved SIM-based secure in the generic model of bilinear maps. In \cite{zhao2018simulation}, Zhao et al.\ present the first SIM-based secure secret-key IPE scheme under the SXDH assumption in the standard model. The authors claim that the scheme can tolerate an unbounded number of ciphertext queries and adaptive key queries. 
Zhao et al. in \cite{zhao2018improved} propose a new version of the scheme, which is an improvement in terms of computational and storage complexity.
In a very recent work \cite{liu2021efficient}, Liu et al.\ present a more efficient and flexible private-key IPE scheme with SIM-based security. To ensure correctness, the scheme requires that the computation of inner products is within a polynomial range, where the discrete logarithm of $g^{\langle \mathbf{x}, \mathbf{y} \rangle}$ can be found in polynomial time. In \cite[Table 2]{liu2021efficient}, the authors compare their proposed IPE scheme with those in \cite{tomida2016efficient,datta2017strongly,zhao2018simulation, zhao2018improved}. The performance of this scheme appears superior in both storage complexity and computation complexity. Moreover,  secret keys and ciphertexts are shorter. 

Although most aforementioned IPE schemes are efficient and based on standard assumptions, they all have one inconvenient property: they are \textit{bounded}. The maximum length of vectors has to be fixed at the beginning, and afterward, one cannot handle vectors whose lengths exceed it. This could be inconvenient when it is hard to predict which data will be encrypted in the setup phase. One may think to solve the problem by setting the maximum length to a  large value. However, the size of parameters expands at least linearly with the fixed maximum length, and such a solution incurs an unnecessary efficiency loss.  In the context of IP-PE and ABE, there exist unbounded schemes (see, for instance, \cite{brakerski2016circuit, chen2018unbounded, lewko2011unbounded, okamoto2012fully}), whose public parameters do not impose a limit on the maximum length of vectors or number of attributes used in the scheme.  In \cite{tomida2020unbounded},  Tomida and Takashima construct two concrete unbounded IPE schemes based on the standard SXDH assumption, both secure in the standard model: the first is a private-key IPE with fully function hiding, the second scheme is a public-key IPE with adaptive security. Concurrently and independently, in \cite{dufour2019unbounded}, Dufour-Sans and Pointcheval describe an unbounded IPE system supporting identity access control with succinct keys. Their construction is proved selectively IND-secure in the random oracle model based on the standard DBDH assumption.
In \cite[Table 1]{tomida2020unbounded}, it is shown a comparison, in terms of efficiency, among private-key schemes that are fully function hiding \cite{datta2016functional, tomida2016efficient, kim2019new} and public-key schemes with adaptive security in the standard model \cite{agrawal2016fully}.

\subsection{Orthogonality FE} A variant of IPE is the so-called \textit{Orthogonality FE} (OFE). In an OFE scheme, a ciphertext $c_{\mathbf{x}}$ is related to a vector $\mathbf{x} \in (\mathbb{Z}_q)^n$, and a functional decryption key $s_{\mathbf{y}}$ is related to a vector $\mathbf{y} \in (\mathbb{Z}_q)^n$. Given $c_{\mathbf{x}}$ and $s_{\mathbf{y}}$, the decryption algorithm determines if $\mathbf{x}$ is orthogonal to $\mathbf{y}$ or not. That is, it returns 1 if $\langle \mathbf{x}, \mathbf{y} \rangle = 0$, and returns 0 otherwise. 
At first glance, OFE appears as a simple modification of IPE, but the latter reveals more information on $\mathbf{x}$ and $\mathbf{y}$ than OFE does. 
Most of the existing OFE schemes are instantiated in (three-factors) composite-order bilinear groups \cite{BW07,KSW08} or DPVS on
prime-order bilinear groups \cite{OT12,okamoto2013efficient}. All of these schemes share a high level of conceptual complexity. In \cite{Wee17}, Wee presents a family of simple OFE schemes in prime-order bilinear groups under the MDDH assumption. In \cite{barbosa2019efficient}, Barbosa et al. extend the Wee's work to the context of function-hiding.

\section{Other FE schemes}\label{chap: other}

So far, we have mainly described schemes for predicate encryption and inner product encryption in their basic form.
In this section, we want to mention other schemes: FE schemes that go beyond predicates and IPE and FE schemes beyond basic features.

\subsection{FE for bounded collusions}\label{sec: bounded collusions}
When the security of an FE scheme is guaranteed only when a constant number of functional decryption keys are corrupted, we talk about {\em security for bounded collusions}.
The case considered in general is the {\em unbounded collusions} one.
Sahai and Seyalioglu \cite{SS10}  build the first FE scheme for all circuits where security handles the corruption of one functional decryption key.
The scheme uses garbles circuits and public-key encryption, and the ciphertext size depends on the size of the circuit associated with the functional decryption keys.
Improvements are presented in \cite{GKP+13} in which the ciphertext size depends only on the size of the output of the function for which the functional decryption keys are generated. This scheme uses ABE for all circuits and FHE, both admitting constructions from standard assumptions.
The works in \cite{GVW12,Agr17} show how to generically turn any FE scheme secure only when one functional decryption key is corrupted into an FE scheme where security handles an (a priori bounded) polynomial number of collusions.

\subsection{FE schemes for quadratic functions}
With the FE scheme for quadratic functions, we intend to allow bilinear  maps computing over the integers.
In particular, messages are expressed as pairs of vectors $(\textbf{x},\textbf{y})\in\Z^n\times\Z^m$ and secret keys are associated with matrices $\textbf{A}=(a_{ij}) \in \Z^{n \times m}$.
The decryption allows to compute $$\textbf{x}^\top\textbf{A}\textbf{y}=\sum_{i,j}a_{ij}x_iy_j.$$
Bilinear maps have several practical applications: for instance, a quadratic polynomial can express different statistical functions (e.g., mean, variance, covariance), the Euclidean distance between two vectors, and the application of a linear or quadratic classifier. FE for an inner product can be used as the building block to obtain FE for quadratic functions. This fact, implicit in \cite{BCFG17}, is made explicit in \cite{gay2020new} and in the private-key variants \cite{lin2017indistinguishability, ananth2017projective}. All the schemes in the works mentioned above have been constructed from standard assumptions on pairings.

In particular, in \cite{BCFG17} Baltico et al.\ present two efficient FE schemes for quadratic functions with linear-size ciphertexts in the public-key setting. One scheme is proved selective IND-secure under standard assumptions (MDDH and 3PDDH), 
the other is proved adaptive IND-secure in the generic-group model. Moreover, in \cite[Table 1]{BCFG17}, its comparison is presented with different schemes for quadratic functions (some of which are in the private-key setting).
Further results have been obtained in \cite{Wee20}, where Wee presents a scheme with constant-size keys and short ciphertexts (selective SIM-based secure on standard assumptions).

\subsection{FE schemes for circuits}
We now consider functionalities of a more general form: arbitrary circuits.
Some results have been obtained for PE schemes for circuits.
In \cite{GVW15}, Gorbunov et al.\ present a leveled PE scheme for all circuits (Boolean predicates of bounded depth), with succinct ciphertexts and secret keys independent of the size of the circuit.
The achieved privacy notion is a selective SIM-based variant of attribute-hiding, assuming the hardness of the subexponential LWE problem.
Recall that the strong variant notion (full attribute-hiding) is impossible to realise for many messages \cite{BSW11,AGVW13}.

Some results have also been obtained for general-purpose FE, in which functions associated with secret keys can be any arbitrary circuits.
In Subsection \ref{sec: bounded collusions} some results for bounded collusions are presented.
In \cite{GGH+16}, Garg et al.\ give constructions for IND obfuscation (based on multilinear maps), and they use it to construct FE for all polynomial-size circuits.
In \cite{BLR+15,GGHZ16}, there are constructions directly based on multilinear maps.
Further analysis can be found in \cite{LT17} where Lin and Tessaro reduce the degree of the required multilinear map to 3.

\subsection{Multi-Input FE} \textit{Multi-Input Functional Encryption} (MI-FE) is a generalisation of functional encryption to the setting of multi-input functions. The functions, evaluated on encrypted information, can be assumed to take multiple inputs, and each input corresponding to a different ciphertext. That is,  the scheme has several encryption slots and each secret key $\mathrm{sk}_f$ for a multi-input function $f$ decrypts jointly ciphertexts $\texttt{Enc}(x_1), \dots, \texttt{Enc}(x_n)$ for all slots to obtain $f(x_1,\dots, x_n)$ without revealing anything more about the encrypted messages. The MI-FE functionality provides the capability to independently encrypt messages for different slots. This concept was introduced by Goldwasser et al.\ in \cite{GGG+14}.  

Abdalla et al.\ in \cite{abdalla2017multi} construct the first scheme of Multi-Input Inner Product Encryption (MI-IPE)  which achieves message privacy, and Datta et al., in \cite{datta2018full}, propose a new scheme which they call \textit{unbounded private-key} MI-IPE. It achieves function-hiding privacy; meanwhile, they enable the encryption of ciphertexts and the generation of secret keys for unbounded vectors. While the inner product functionality is helpful for several meaningful applications (we refer the reader to \cite{abdalla2017multi} for a discussion), it is evidently desirable, from the viewpoint of both theory and practice, to extend the reach of MI-FE from standard assumptions beyond inner products. In \cite{abdalla2018multi}, Abdalla et al. put forward a novel methodology to convert single-input FE for inner
products into multi-input schemes for the same functionality. The resulting scheme does not require pairings, it can be instantiated with all
known single-input schemes and it relies on different assumptions, such as DDH, LWE, and DCR.
In a very recent preprint, \cite{cryptoeprint:2020:1285}, Agrawal et al.\ present the first MI-FE scheme for quadratic functions from pairings. All previous MI-FE schemes either support only inner products (linear functions) or rely on solid cryptographic assumptions such as indistinguishability obfuscation or multilinear maps. The scheme showed in \cite{cryptoeprint:2020:1285} would achieve (selective) IND-based security against unbounded collusions under the standard BMDDH assumption. Thus, their construction would accomplish the same level of protection as single input quadratic FE.

In \cite{ACF+20}, the notion of \textit{Ad Hoc MI-FE} (aMI-FE) was introduced to handle the key escrow problem of MI-FE. aMI-FE does not differentiate between key authorities and users, and it lets users generate their own partial decryption keys along with ciphertexts. In the same work, Agrawal et al.\ provide construction of aMI-FE for inner product based on the LWE assumption.

Recently, Agrawal, Goyal, and Tomida in \cite{AGT20} define
\textit{Multi-Party Functional Encryption} (MP-FE), which extends and abstracts various notions of functional encryption such as multi-input, multi-client, multi-authority, and dynamic decentralised that we present in the following sections. 
MP-FE allows for distributed ciphertexts and distributed keys and specifies how these may be combined for function evaluation. It comprehends both key-policy and ciphertext-policy schemes, and it captures both attribute and function hiding. It also supports schemes with interactive, independent, or centralised setup.
The authors individuate several combinations of features not studied in literature and provide constructions for some of them with their general framework.

\subsection{Multi-Client FE}\label{sec: Multi client}
\textit{Multi-Client Functional Encryption} (MC-FE), as defined in \cite{GGG+14}, generalises the definition of functional encryption in the following way. The single input $x\in X$ to the encryption procedure is divided into $n$ independent components $(x_1,\ldots,x_n)$.
Each client $C_i$ does the following encryption  $\texttt{Enc}({\rm pk},{\rm ek}_i,x_i,\ell)$, where $\ell$ is a label, usually time-dependent, and ${\rm ek}_i$ is an encryption key specific for the client $C_i$.
A user owning a secret key ${\rm sk}_f$ for an $n$-ary function $f$ and $n$ ciphertexts  $\texttt{Enc}({\rm pk},{\rm ek}_1,x_1,\ell)$, \ldots, $\texttt{Enc}({\rm pk},{\rm ek}_n,x_n,\ell)$ can compute $f(x_1,\ldots,x_n)$ but he will learn nothing else about each $x_i$.
The difference with Multi-Input FE is that the combination of ciphertexts generated for different labels (hence at other times) does not give a valid global ciphertext, and the adversary learns nothing from it. In MI-FE, every ciphertext from every slot can be combined with any other ciphertext from any different slot.
The first MC-FE for an inner product in the standard model (under LWE assumption) can be found in \cite{LT19}. Other MC-FE schemes can be found in \cite{CSW20,ABM+20}.

\subsection{Multi-Authority and Decentralised FE}\label{sec: multi-authority} 
A drawback of standard FE is that a single party, the authority, is responsible for creating the functional decryption keys for all users in the system. As a direct consequence, this authority can decrypt all messages since the authority has to create every possible decryption key. Thus, relying on a single authority has consequences for the scalability of the system and the trust relations. In practical situations, we would instead appoint multiple authorities, where each authority is responsible for issuing keys but cannot decrypt messages unless it cooperates with (all) other authorities. The question of whether it is possible to construct such a \textit{multi-authority scheme} was first raised by Sahai and Waters \cite{SW05}. Whereas in \cite{lewko2011decentralizing}, Lewko and Waters address both trust and scalability issues rather: no authority is required to hold a master key, and new authorities can be added to the system without requiring any form of interaction. The resulting scheme is so-called a \textit{decentralised} FE scheme.

\subsubsection{\bf Multi-Authority FE} In a \textit{Multi-Authority Predicate Encryption} (MA-PE) system, ciphertexts are associated with one or more predicates from various authorities.  
Only users with keys satisfying all these predicates can decrypt the ciphertext.
The first proposed MA-PE schemes either require interaction between all authorities \cite{Cha07,muller2008distributed}, or do not allow for the addition of new authorities once the system is set up \cite{CC09}. They are all multi-authority attribute-based encryption (MA-ABE), and they are based on the DBDH assumption. Further results on MA-ABE scheme can be found in \cite{LMS15,LMS16,li2020decentralized}.  
The \textit{Controlled Functional Encryption} (C-FE), introduced by Naveed
et al.\ \cite{naveed2014controlled}, is a special flavor of FE.
In C-FE, the decryption algorithm becomes a two-party protocol between the party holding the function key and the authority, where only the former learns the output of the computation. In a very recent work \cite{ambrona2021controlled}, Ambrona et al.\ extend  C-FE to a multi-authority C-FE by distributing the role of (and the trust in) the authority across several parties.

\subsubsection{\bf Decentralised FE}

As above-mentioned, the first FE scheme that does not require a trusted authority is due to Lewko and Waters, \cite{lewko2011decentralizing}. It is a decentralised MA-ABE, where each authority works completely separately, and failures or disruptions for some authorities will not affect other authorities. This makes the system stronger than the majority of MA-ABE schemes. Apart from the initial set of public parameters built by all authorities, authorities no longer need to establish trustable relationships with each other in advance. 
Other decentralised MA-ABE can be found in \cite{lin2010secure, CC09, rahulamathavan2015user,yang2018improving}. In all these works,  collusion resistance is obtained using a \textit{Global Identifier} (GID), but this can breach the user privacy. In \cite{yang2018improving}, this issue is studied: unique user identifiers are obtained by combining a user’s identity with the identity of the (local) attribute authority (AA) where the user is located. The user identities remain private to the AA outside the domain, which enhances privacy and security. The scheme is based on composite-order bilinear groups and some variants of SD3, whereas security proof makes use of the dual system encryption methodology (for more details, see \cite{Wat09} where this methodology is first introduced).

In \cite{van2020multi}, van de Kamp et al.\ propose a generic framework for creating decentralised multi-authority predicate encryption. The framework supports several predicate types, such as MA-IBE, MA-ABE, and MA-IP-PE. In the random oracle model, the resulting encryption schemes are proved fully secure under standard subgroup decision assumptions, formulated for a bilinear group of orders a product of 3 primes.

The definition of MC-FE stated in Subsection \ref{sec: Multi client} assumes the existence of a trusted third party who runs the \texttt{Setup} algorithm and distributes the functional decryption keys. As introduced in \cite{CDG+18},
in a Decentralised MC-FE (DMC-FE), the authority is removed, and the clients work together to generate appropriate functional decryption keys. In \cite{CDG+18}, Chotard et al.\ give the first DMC-FE from standard assumptions for inner products. Security is proved using bilinear pairing groups, and the scheme handles the corruption of input slots. In \cite{LT19}, Libert and {\c{T}}i{\c{t}}iu present the first decentralised MC-FE scheme based on LWE assumptions. In \cite{abdalla2019decentralizing}, Abdalla et al.\ provide a generic compiler from any MC-FE scheme satisfying an extra property, called \textit{special key derivation}\footnote{In a MI-FE or MC-FE scheme with the \textit{special key derivation} property the master secret key can be split into separate secret keys, one for each input, i.e. $\rm mk = \{{\rm mk}_i\}_{i=1,\dots,n}$, and functional decryption keys ${\rm sk}_f$ are derived through a combination of local and linear inner-product computations on $i, f, {\rm sk}_i$ and $\rm pp$.}, into a DMC-FE scheme. The transformation is purely information-theoretic, and it does not require any additional assumptions. As the MC-FE from Chotard et al. \cite{CDG+18} satisfies this extra property, the resulting scheme is a DMC-FE scheme secure under the basic DDH assumption without pairings (in the random oracle model). As in \cite{CDG+18}, the version of the scheme without labels is secure in the standard model.

In \cite{chotard2020dynamic}, Chotard et al.\ introduce the \textit{Dynamic Decentralised Functional Encryption} (DD-FE), a generalisation of FE which allows multiple users to join the system dynamically, without relying on a trusted third party or expensive and interactive Multi-Party Computation protocols. This notion subsumes existing multi-user extensions of FE, such as Multi-Input, Multi-Client, and Multi-Party FE.

\subsection{Hierarchical FE} Motivated by the applicability of functional encryption to expressive access control systems, the \textit{Hierarchical Functional Encryption} has been introduced by Ananth et al.\ in \cite{ananth2013differing}. In a hierarchical FE scheme, the holder of any functional decryption key ${\rm sk}_f$ for a function $f$ can in turn generate a functional decryption key ${\rm sk}_{g \circ f}$ corresponding to the function $g \circ f$ for any given function $g$. Now,  anyone holding the delegated functional key ${\rm sk}_{g \circ f}$ and an encryption of any message $x$ can compute $g(f(x))$, but cannot learn any additional information about the message $x$ and $f(x)$. Such expressive delegation capabilities give rise to hierarchical access-control, a sought-after ingredient in modern access control systems. In particular, the notion of hierarchical FE generalises those of HIBE and HABE (see Paragraph \ref{subsec: HIBE} and Paragraph \ref{subsec: HABE}). In \cite{brakerski2017hierarchical}, Brakerski et al.\ present a generic transformation that converts any general-purpose public-key FE scheme into a hierarchical one without relying on any additional assumptions. This construction yields various hierarchical schemes based on various assumptions in the standard model, such as multilinear maps and the LWE assumption.

\subsection{Traceable FE} The conventional definition of FE associates each function to a secret functional decryption key, and therefore, all users get the same private key for the same function. This induces an important problem: if one of these users (called a \textit{traitor}) leaks or sells a  functional decryption key to be included in a pirate decryption tool, then there is no way to trace back its identity.
\textit{Traitor tracing} is a mechanism enabling an authority or an arbitrary party (a delegated party in the system with this specific task) to identify malicious users who possibly colluded.  
Several IBE \cite{abdalla2007identity,goyal2007reducing,au2008traceable,phan2011identity} and ABE \cite{ning2014large,liu2013blackbox,liu2015practical,liu2015traceable} schemes take into consideration the traceability. Indeed, achieving traceability is usually very expensive. Adding traitor tracing to public-key encryption requires a very high extra cost: even in the bounded model, the cost grows proportionally with the number of traitors. In \cite{do2020traceable}, a new primitive, called \textit{Traceable Functional Encryption}, is defined. 
The functional decryption key will not only be specific to a function but to a user too, in such a way that if some users collude to produce a pirate decoder that successfully evaluates a function on the plaintext, from the ciphertext only, one can trace back at least one of them. In \cite{do2020traceable}, Do et al.\ propose a traceable IP-FE with black-box confirmation, using the pairing.

\subsection{Verifiable FE}
In the standard setting of FE, authority and the encryptors are supposed to run their respective algorithms faithfully. In \cite{tang2010verifiable}, Tang and Ji introduce the notion of \textit{verifiable ABE} (VABE), where it is checked whether the functional decryption keys have been generated correctly. The core idea of the method is using some trick to change secret sharing \cite{shamir1979share} in the ABE schemes with verifiable secret sharing \cite{tang2004non}. In \cite{badrinarayanan2016verifiable}, Badrinarayanan et al.\ extend the VABE to the concept of \textit{verifiable FE} (VFE), which essentially guarantees that dishonest encryptors and authorities, even when colluding together, are not able to generate ciphertexts and functional decryption keys that give “inconsistent” results.  But, the solution proposed in \cite{badrinarayanan2016verifiable} is not applicable for the IPE scheme since it requires a perfectly correct IPE scheme. For this reason, in \cite{soroush2020verifiable}, Soroush et al.\ improve (in terms of efficiency) the construction in \cite{badrinarayanan2016verifiable} and they construct the first
efficient perfectly correct IPE scheme, whose security relies on the DLIN assumption.

\section{Conclusions}\label{chap: concl}

In this paper, we summarise the state-of-the-art for Functional Encryption schemes.
We provide an overview of the main mathematical assumptions used for proving the security of FE schemes and explicitly exhibit some connections between these assumptions.
We also give an overview of the used security models.
For IBE, ABE, Predicate Encryption with private index, and Inner Product Encryption, we describe the basic structure of the schemes and present their main characteristics.
For each class of FE schemes, we provide several references to the reader, focusing on the works that first offered a practical construction and the most recent constructions.
Finally, we present schemes that go beyond the inner product and schemes with enhanced features.
The latter comprehend schemes useful in real-life scenarios.

In this survey, we give great relevance to the mathematical assumptions which many FE schemes rely upon. We observe that there are assumptions not quantum-resistant or hardly realisable at the basis of many constructions. The latter means that there are few explicit examples of algebraic structures which satisfy the given assumption.  
For this reason, in our opinion, more work should be done to obtain solutions which are safer and more practical.

Even though FE schemes can be helpful for different applications and different real scenarios, as described in Section \ref{chap: IPE} and \ref{chap: other}, to our knowledge, the only practical implementations of FE schemes belong to the class of IBE and ABE.
These classes are undoubtedly the most studied because they are chronologically the first to appear and because many works focus on efficient constructions of these schemes.
For the other classes, to our knowledge, we can only find some theoretical analyses of the systems, but actual implementations are still missing.
Indeed, in scenarios where there is the need to learn a function of the encrypted data, FHE or similar schemes are so far preferred.

Nevertheless, since FE schemes have great potential, we believe that their research will still be flourishing in the following years.

\section*{Acknowledgment}
This survey has been partially presented at the CifrisCloud seminar series on March 2021 \cite{DeCifris}.
We would like to thank Andrea Gelpi, Marco Pedicini, and Michela Iezzi for their interest and valuable discussion on the subject.
We also thank the anonymous referee for the helpful comments.

\bibliographystyle{acm}
\bibliography{References,References2}

\begin{thebibliography}{100}

\bibitem{abdalla2019decentralizing}
{\sc Abdalla, M., Benhamouda, F., Kohlweiss, M., and Waldner, H.}
\newblock Decentralizing inner-product functional encryption.
\newblock In {\em IACR International Workshop on Public Key Cryptography\/}
  (2019), Springer, pp.~128--157.

\bibitem{ABDP15}
{\sc Abdalla, M., Bourse, F., De~Caro, A., and Pointcheval, D.}
\newblock Simple functional encryption schemes for inner products.
\newblock In {\em IACR International Workshop on Public Key Cryptography\/}
  (2015), Springer, pp.~733--751.

\bibitem{ABM+20}
{\sc Abdalla, M., Bourse, F., Marival, H., Pointcheval, D., Soleimanian, A.,
  and Waldner, H.}
\newblock Multi-client inner-product functional encryption in the random-oracle
  model.
\newblock In {\em International Conference on Security and Cryptography for
  Networks\/} (2020), Springer, pp.~525--545.

\bibitem{abdalla2018multi}
{\sc Abdalla, M., Catalano, D., Fiore, D., Gay, R., and Ursu, B.}
\newblock Multi-input functional encryption for inner products:
  {Function}-hiding realizations and constructions without pairings.
\newblock In {\em Annual International Cryptology Conference\/} (2018),
  Springer, pp.~597--627.

\bibitem{abdalla2007identity}
{\sc Abdalla, M., Dent, A.~W., Malone-Lee, J., Neven, G., Phan, D.~H., and
  Smart, N.~P.}
\newblock Identity-based traitor tracing.
\newblock In {\em International Workshop on Public Key Cryptography\/} (2007),
  Springer, pp.~361--376.

\bibitem{abdalla2017multi}
{\sc Abdalla, M., Gay, R., Raykova, M., and Wee, H.}
\newblock Multi-input inner-product functional encryption from pairings.
\newblock In {\em Annual International Conference on the Theory and
  Applications of Cryptographic Techniques\/} (2017), Springer, pp.~601--626.

\bibitem{Agr17}
{\sc Agrawal, S.}
\newblock Stronger security for reusable garbled circuits, general definitions
  and attacks.
\newblock In {\em Annual International Cryptology Conference\/} (2017),
  Springer, pp.~3--35.

\bibitem{agrawal2015practical}
{\sc Agrawal, S., Agrawal, S., Badrinarayanan, S., Kumarasubramanian, A.,
  Prabhakaran, M., and Sahai, A.}
\newblock On the practical security of inner product functional encryption.
\newblock In {\em IACR International Workshop on Public Key Cryptography\/}
  (2015), Springer, pp.~777--798.

\bibitem{ABB10}
{\sc Agrawal, S., Boneh, D., and Boyen, X.}
\newblock Efficient lattice {(H) IBE} in the standard model.
\newblock In {\em Annual International Conference on the Theory and
  Applications of Cryptographic Techniques\/} (2010), Springer, pp.~553--572.

\bibitem{ACF+20}
{\sc Agrawal, S., Clear, M., Frieder, O., Garg, S., O'Neill, A., and Thaler,
  J.}
\newblock Ad hoc multi-input functional encryption.
\newblock In {\em 11th Innovations in Theoretical Computer Science Conference
  (ITCS 2020)\/} (2020), Schloss Dagstuhl-Leibniz-Zentrum f{\"u}r Informatik.

\bibitem{AFV11}
{\sc Agrawal, S., Freeman, D.~M., and Vaikuntanathan, V.}
\newblock Functional encryption for inner product predicates from learning with
  errors.
\newblock In {\em International Conference on the Theory and Application of
  Cryptology and Information Security\/} (2011), Springer, pp.~21--40.

\bibitem{AGVW13}
{\sc Agrawal, S., Gorbunov, S., Vaikuntanathan, V., and Wee, H.}
\newblock Functional encryption: {New} perspectives and lower bounds.
\newblock In {\em Annual Cryptology Conference\/} (2013), Springer,
  pp.~500--518.

\bibitem{cryptoeprint:2020:1285}
{\sc Agrawal, S., Goyal, R., and Tomida, J.}
\newblock Multi-input quadratic functional encryption from pairings.
\newblock Cryptology ePrint Archive, Report 2020/1285, 2020.
\newblock \url{https://eprint.iacr.org/2020/1285}.

\bibitem{AGT20}
{\sc Agrawal, S., Goyal, R., and Tomida, J.}
\newblock Multi-party functional encryption.
\newblock Cryptology ePrint Archive, Report 2020/1266, 2020.
\newblock \url{https://eprint.iacr.org/2020/1266}.

\bibitem{agrawal2020adaptive}
{\sc Agrawal, S., Libert, B., Maitra, M., and Titiu, R.}
\newblock Adaptive simulation security for inner product functional encryption.
\newblock In {\em IACR International Conference on Public-Key Cryptography\/}
  (2020), Springer, pp.~34--64.

\bibitem{agrawal2016fully}
{\sc Agrawal, S., Libert, B., and Stehl{\'e}, D.}
\newblock Fully secure functional encryption for inner products, from standard
  assumptions.
\newblock In {\em Annual International Cryptology Conference\/} (2016),
  Springer, pp.~333--362.

\bibitem{AS+19}
{\sc Al-Dahhan, R.~R., Shi, Q., Lee, G.~M., and Kifayat, K.}
\newblock Survey on revocation in ciphertext-policy attribute-based encryption.
\newblock {\em Sensors 19}, 7 (2019), 1695.

\bibitem{alwen2013relationship}
{\sc Alwen, J., Barbosa, M., Farshim, P., Gennaro, R., Gordon, S.~D., Tessaro,
  S., and Wilson, D.~A.}
\newblock On the relationship between functional encryption, obfuscation, and
  fully homomorphic encryption.
\newblock In {\em IMA International Conference on Cryptography and Coding\/}
  (2013), Springer, pp.~65--84.

\bibitem{ambrona2021controlled}
{\sc Ambrona, M., Fiore, D., and Soriente, C.}
\newblock Controlled functional encryption revisited: {Multi}-authority
  extensions and efficient schemes for quadratic functions.
\newblock {\em Proceedings on Privacy Enhancing Technologies 2021}, 1 (2021),
  21--42.

\bibitem{ananth2013differing}
{\sc Ananth, P., Boneh, D., Garg, S., Sahai, A., and Zhandry, M.}
\newblock Differing-inputs obfuscation and applications.
\newblock Cryptology ePrint Archive, Report 2013/689, 2013.
\newblock \url{https://eprint.iacr.org/2013/689}.

\bibitem{ABSV15}
{\sc Ananth, P., Brakerski, Z., Segev, G., and Vaikuntanathan, V.}
\newblock From selective to adaptive security in functional encryption.
\newblock In {\em Annual Cryptology Conference\/} (2015), Springer,
  pp.~657--677.

\bibitem{ananth2017projective}
{\sc Ananth, P., and Sahai, A.}
\newblock Projective arithmetic functional encryption and indistinguishability
  obfuscation from degree-5 multilinear maps.
\newblock In {\em Annual International Conference on the Theory and
  Applications of Cryptographic Techniques\/} (2017), Springer, pp.~152--181.

\bibitem{AHL+12}
{\sc Attrapadung, N., Herranz, J., Laguillaumie, F., Libert, B., De~Panafieu,
  E., and R{\`a}fols, C.}
\newblock Attribute-based encryption schemes with constant-size ciphertexts.
\newblock {\em Theoretical computer science 422\/} (2012), 15--38.

\bibitem{au2008traceable}
{\sc Au, M.~H., Huang, Q., Liu, J.~K., Susilo, W., Wong, D.~S., and Yang, G.}
\newblock Traceable and retrievable identity-based encryption.
\newblock In {\em International Conference on Applied Cryptography and Network
  Security\/} (2008), Springer, pp.~94--110.

\bibitem{badrinarayanan2016verifiable}
{\sc Badrinarayanan, S., Goyal, V., Jain, A., and Sahai, A.}
\newblock Verifiable functional encryption.
\newblock In {\em International Conference on the Theory and Application of
  Cryptology and Information Security\/} (2016), Springer, pp.~557--587.

\bibitem{BCFG17}
{\sc Baltico, C. E.~Z., Catalano, D., Fiore, D., and Gay, R.}
\newblock Practical functional encryption for quadratic functions with
  applications to predicate encryption.
\newblock In {\em Annual International Cryptology Conference\/} (2017),
  Springer, pp.~67--98.

\bibitem{barbosa2019efficient}
{\sc Barbosa, M., Catalano, D., Soleimanian, A., and Warinschi, B.}
\newblock Efficient function-hiding functional encryption: {From}
  inner-products to orthogonality.
\newblock In {\em Cryptographers’ Track at the RSA Conference\/} (2019),
  Springer, pp.~127--148.

\bibitem{BC+19}
{\sc Bartusek, J., Carmer, B., Jain, A., Jin, Z., Lepoint, T., Ma, F., Malkin,
  T., Malozemoff, A.~J., and Raykova, M.}
\newblock Public-key function-private hidden vector encryption (and more).
\newblock In {\em International Conference on the Theory and Application of
  Cryptology and Information Security\/} (2019), Springer, pp.~489--519.

\bibitem{BR93}
{\sc Bellare, M., and Rogaway, P.}
\newblock Random oracles are practical: {A} paradigm for designing efficient
  protocols.
\newblock In {\em Proceedings of the 1st ACM Conference on Computer and
  Communications Security\/} (1993), pp.~62--73.

\bibitem{benhamouda2017cca}
{\sc Benhamouda, F., Bourse, F., and Lipmaa, H.}
\newblock {CCA}-secure inner-product functional encryption from projective hash
  functions.
\newblock In {\em IACR International Workshop on Public Key Cryptography\/}
  (2017), Springer, pp.~36--66.

\bibitem{BSW07}
{\sc Bethencourt, J., Sahai, A., and Waters, B.}
\newblock Ciphertext-policy attribute-based encryption.
\newblock In {\em 2007 IEEE symposium on security and privacy (SP'07)\/}
  (2007), IEEE, pp.~321--334.

\bibitem{bishop2015function}
{\sc Bishop, A., Jain, A., and Kowalczyk, L.}
\newblock Function-hiding inner product encryption.
\newblock In {\em International Conference on the Theory and Application of
  Cryptology and Information Security\/} (2015), Springer, pp.~470--491.

\bibitem{blake2006refinements}
{\sc Blake, I.~F., Murty, V.~K., and Xu, G.}
\newblock Refinements of {Miller}'s algorithm for computing the {Weil/Tate}
  pairing.
\newblock {\em Journal of Algorithms 58}, 2 (2006), 134--149.

\bibitem{BBP19}
{\sc Blazy, O., Brouilhet, L., and Phan, D.~H.}
\newblock Anonymous identity based encryption with traceable identities.
\newblock In {\em Proceedings of the 14th International Conference on
  Availability, Reliability and Security\/} (2019), pp.~1--10.

\bibitem{BB04}
{\sc Boneh, D., and Boyen, X.}
\newblock Efficient selective-{ID} secure identity-based encryption without
  random oracles.
\newblock In {\em International conference on the theory and applications of
  cryptographic techniques\/} (2004), Springer, pp.~223--238.

\bibitem{BB04b}
{\sc Boneh, D., and Boyen, X.}
\newblock Secure identity based encryption without random oracles.
\newblock In {\em Annual International Cryptology Conference\/} (2004),
  Springer, pp.~443--459.

\bibitem{BBG05}
{\sc Boneh, D., Boyen, X., and Goh, E.-J.}
\newblock Hierarchical identity based encryption with constant size ciphertext.
\newblock In {\em Annual International Conference on the Theory and
  Applications of Cryptographic Techniques\/} (2005), Springer, pp.~440--456.

\bibitem{BBS04}
{\sc Boneh, D., Boyen, X., and Shacham, H.}
\newblock Short group signatures.
\newblock In {\em Annual international cryptology conference\/} (2004),
  Springer, pp.~41--55.

\bibitem{BCOP04}
{\sc Boneh, D., Di~Crescenzo, G., Ostrovsky, R., and Persiano, G.}
\newblock Public key encryption with keyword search.
\newblock In {\em International conference on the theory and applications of
  cryptographic techniques\/} (2004), Springer, pp.~506--522.

\bibitem{BF01}
{\sc Boneh, D., and Franklin, M.}
\newblock Identity-based encryption from the {Weil} pairing.
\newblock In {\em Annual international cryptology conference\/} (2001),
  Springer, pp.~213--229.

\bibitem{BGG+14}
{\sc Boneh, D., Gentry, C., Gorbunov, S., Halevi, S., Nikolaenko, V., Segev,
  G., Vaikuntanathan, V., and Vinayagamurthy, D.}
\newblock Fully key-homomorphic encryption, arithmetic circuit {ABE} and
  compact garbled circuits.
\newblock In {\em Annual International Conference on the Theory and
  Applications of Cryptographic Techniques\/} (2014), Springer, pp.~533--556.

\bibitem{BGH07}
{\sc Boneh, D., Gentry, C., and Hamburg, M.}
\newblock Space-efficient identity based encryptionwithout pairings.
\newblock In {\em 48th Annual IEEE Symposium on Foundations of Computer Science
  (FOCS'07)\/} (2007), IEEE, pp.~647--657.

\bibitem{BLR+15}
{\sc Boneh, D., Lewi, K., Raykova, M., Sahai, A., Zhandry, M., and Zimmerman,
  J.}
\newblock Semantically secure order-revealing encryption: {Multi}-input
  functional encryption without obfuscation.
\newblock In {\em Annual International Conference on the Theory and
  Applications of Cryptographic Techniques\/} (2015), Springer, pp.~563--594.

\bibitem{BRS13a}
{\sc Boneh, D., Raghunathan, A., and Segev, G.}
\newblock Function-private identity-based encryption: {Hiding} the function in
  functional encryption.
\newblock In {\em Annual Cryptology Conference\/} (2013), Springer,
  pp.~461--478.

\bibitem{BRS13b}
{\sc Boneh, D., Raghunathan, A., and Segev, G.}
\newblock Function-private subspace-membership encryption and its applications.
\newblock In {\em International Conference on the Theory and Application of
  Cryptology and Information Security\/} (2013), Springer, pp.~255--275.

\bibitem{boneh2006fully}
{\sc Boneh, D., Sahai, A., and Waters, B.}
\newblock Fully collusion resistant traitor tracing with short ciphertexts and
  private keys.
\newblock In {\em Annual International Conference on the Theory and
  Applications of Cryptographic Techniques\/} (2006), Springer, pp.~573--592.

\bibitem{BSW11}
{\sc Boneh, D., Sahai, A., and Waters, B.}
\newblock Functional encryption: {Definitions} and challenges.
\newblock In {\em Theory of Cryptography Conference\/} (2011), Springer,
  pp.~253--273.

\bibitem{BSW12}
{\sc Boneh, D., Sahai, A., and Waters, B.}
\newblock Functional encryption: {A} new vision for public-key cryptography.
\newblock {\em Communications of the ACM 55}, 11 (2012), 56--64.

\bibitem{BS03}
{\sc Boneh, D., and Silverberg, A.}
\newblock Applications of multilinear forms to cryptography.
\newblock {\em Contemporary Mathematics 324}, 1 (2003), 71--90.

\bibitem{BW07}
{\sc Boneh, D., and Waters, B.}
\newblock Conjunctive, subset, and range queries on encrypted data.
\newblock In {\em Theory of cryptography conference\/} (2007), Springer,
  pp.~535--554.

\bibitem{Boy03}
{\sc Boyen, X.}
\newblock Multipurpose identity-based signcryption.
\newblock In {\em Annual International Cryptology Conference\/} (2003),
  Springer, pp.~383--399.

\bibitem{Boy08}
{\sc Boyen, X.}
\newblock A tapestry of identity-based encryption: {Practical} frameworks
  compared.
\newblock {\em International Journal of Applied Cryptography 1}, 1 (2008),
  3--21.

\bibitem{BW06}
{\sc Boyen, X., and Waters, B.}
\newblock Anonymous hierarchical identity-based encryption (without random
  oracles).
\newblock In {\em Annual International Cryptology Conference\/} (2006),
  Springer, pp.~290--307.

\bibitem{brakerski2017hierarchical}
{\sc Brakerski, Z., Chandran, N., Goyal, V., Jain, A., Sahai, A., and Segev,
  G.}
\newblock Hierarchical functional encryption.
\newblock In {\em 8th Innovations in Theoretical Computer Science Conference
  (ITCS 2017)\/} (2017), Schloss Dagstuhl-Leibniz-Zentrum fuer Informatik.

\bibitem{brakerski2018function}
{\sc Brakerski, Z., and Segev, G.}
\newblock Function-private functional encryption in the private-key setting.
\newblock {\em Journal of Cryptology 31}, 1 (2018), 202--225.

\bibitem{brakerski2016circuit}
{\sc Brakerski, Z., and Vaikuntanathan, V.}
\newblock Circuit-{ABE} from {LWE}: {Unbounded} attributes and semi-adaptive
  security.
\newblock In {\em Annual International Cryptology Conference\/} (2016),
  Springer, pp.~363--384.

\bibitem{CK+09}
{\sc Camenisch, J., Kohlweiss, M., Rial, A., and Sheedy, C.}
\newblock Blind and anonymous identity-based encryption and authorised private
  searches on public key encrypted data.
\newblock In {\em International workshop on public key cryptography\/} (2009),
  Springer, pp.~196--214.

\bibitem{CHK03}
{\sc Canetti, R., Halevi, S., and Katz, J.}
\newblock A forward-secure public-key encryption scheme.
\newblock In {\em International Conference on the Theory and Applications of
  Cryptographic Techniques\/} (2003), Springer, pp.~255--271.

\bibitem{CHKP10}
{\sc Cash, D., Hofheinz, D., Kiltz, E., and Peikert, C.}
\newblock Bonsai trees, or how to delegate a lattice basis.
\newblock In {\em Annual international conference on the theory and
  applications of cryptographic techniques\/} (2010), Springer, pp.~523--552.

\bibitem{castagnos2018practical}
{\sc Castagnos, G., Laguillaumie, F., and Tucker, I.}
\newblock Practical fully secure unrestricted inner product functional
  encryption modulo $p$.
\newblock In {\em International Conference on the Theory and Application of
  Cryptology and Information Security\/} (2018), Springer, pp.~733--764.

\bibitem{Cha07}
{\sc Chase, M.}
\newblock Multi-authority attribute based encryption.
\newblock In {\em Theory of cryptography conference\/} (2007), Springer,
  pp.~515--534.

\bibitem{CC09}
{\sc Chase, M., and Chow, S.~S.}
\newblock Improving privacy and security in multi-authority attribute-based
  encryption.
\newblock In {\em Proceedings of the 16th ACM conference on Computer and
  communications security\/} (2009), pp.~121--130.

\bibitem{CS05}
{\sc Chatterjee, S., and Sarkar, P.}
\newblock Trading time for space: {Towards} an efficient {IBE} scheme with
  short(er) public parameters in the standard model.
\newblock In {\em International Conference on Information Security and
  Cryptology\/} (2005), Springer, pp.~424--440.

\bibitem{chen2018unbounded}
{\sc Chen, J., Gong, J., Kowalczyk, L., and Wee, H.}
\newblock Unbounded {ABE} via bilinear entropy expansion, revisited.
\newblock In {\em Annual International Conference on the Theory and
  Applications of Cryptographic Techniques\/} (2018), Springer, pp.~503--534.

\bibitem{chen2019identity}
{\sc Chen, J., Ling, J., Ning, J., and Ding, J.}
\newblock Identity-based signature schemes for multivariate public key
  cryptosystems.
\newblock {\em The Computer Journal 62}, 8 (2019), 1132--1147.

\bibitem{CW13}
{\sc Chen, J., and Wee, H.}
\newblock Fully,(almost) tightly secure {IBE} and dual system groups.
\newblock In {\em Annual Cryptology Conference\/} (2013), Springer,
  pp.~435--460.

\bibitem{cheon2015cryptanalysis}
{\sc Cheon, J.~H., Han, K., Lee, C., Ryu, H., and Stehl{\'e}, D.}
\newblock Cryptanalysis of the multilinear map over the integers.
\newblock In {\em Annual International Conference on the Theory and
  Applications of Cryptographic Techniques\/} (2015), Springer, pp.~3--12.

\bibitem{chotard2020dynamic}
{\sc Chotard, J., Dufour-Sans, E., Gay, R., Phan, D.~H., and Pointcheval, D.}
\newblock Dynamic decentralized functional encryption.
\newblock In {\em Annual International Cryptology Conference\/} (2020),
  Springer, pp.~747--775.

\bibitem{CDG+18}
{\sc Chotard, J., Sans, E.~D., Gay, R., Phan, D.~H., and Pointcheval, D.}
\newblock Decentralized multi-client functional encryption for inner product.
\newblock In {\em International Conference on the Theory and Application of
  Cryptology and Information Security\/} (2018), Springer, pp.~703--732.

\bibitem{CSW20}
{\sc Ciampi, M., Siniscalchi, L., and Waldner, H.}
\newblock Multi-client functional encryption for separable functions.
\newblock In {\em Public-Key Cryptography - {PKC} 2021 - 24th {IACR}
  International Conference on Practice and Theory of Public Key Cryptography,
  Virtual Event, May 10-13, 2021, Proceedings, Part {I}\/} (2021), vol.~12710
  of {\em Lecture Notes in Computer Science}, Springer, pp.~724--753.

\bibitem{Coc01}
{\sc Cocks, C.}
\newblock An identity based encryption scheme based on quadratic residues.
\newblock In {\em IMA international conference on cryptography and coding\/}
  (2001), Springer, pp.~360--363.

\bibitem{coron2013practical}
{\sc Coron, J.-S., Lepoint, T., and Tibouchi, M.}
\newblock Practical multilinear maps over the integers.
\newblock In {\em Annual Cryptology Conference\/} (2013), Springer,
  pp.~476--493.

\bibitem{CH+20}
{\sc Cui, H., Hon~Yuen, T., Deng, R.~H., and Wang, G.}
\newblock Server-aided revocable attribute-based encryption for cloud computing
  services.
\newblock {\em Concurrency and Computation: Practice and Experience 32}, 14
  (2020), e5680.

\bibitem{DRS17}
{\sc Daniel, R.~M., Rajsingh, E.~B., and Silas, S.}
\newblock Analysis of hierarchical identity based encryption schemes and its
  applicability to computing environments.
\newblock {\em Journal of information security and applications 36\/} (2017),
  20--31.

\bibitem{datta2016functional}
{\sc Datta, P., Dutta, R., and Mukhopadhyay, S.}
\newblock Functional encryption for inner product with full function privacy.
\newblock In {\em Public-Key Cryptography--PKC 2016}. Springer, 2016,
  pp.~164--195.

\bibitem{datta2017strongly}
{\sc Datta, P., Dutta, R., and Mukhopadhyay, S.}
\newblock Strongly full-hiding inner product encryption.
\newblock {\em Theoretical Computer Science 667\/} (2017), 16--50.

\bibitem{datta2018full}
{\sc Datta, P., Okamoto, T., and Tomida, J.}
\newblock Full-hiding (unbounded) multi-input inner product functional
  encryption from the $k$-linear assumption.
\newblock In {\em IACR International Workshop on Public Key Cryptography\/}
  (2018), Springer, pp.~245--277.

\bibitem{DIP12}
{\sc De~Caro, A., Iovino, V., and Persiano, G.}
\newblock Fully secure hidden vector encryption.
\newblock In {\em International Conference on Pairing-Based Cryptography\/}
  (2012), Springer, pp.~102--121.

\bibitem{DeCifris}
{\sc {De Componendis Cifris}}.
\newblock Functional {E}ncryption, an overview - {C}arla {M}ascia, {I}rene
  {V}illa.
\newblock \url{https://www.youtube.com/watch?v=jz8v22jDlAs},
  \url{https://www.decifris.it/cifrisCloud}.
\newblock Accessed: 2021-08-12.

\bibitem{DW+14}
{\sc Deng, H., Wu, Q., Qin, B., Domingo-Ferrer, J., Zhang, L., Liu, J., and
  Shi, W.}
\newblock Ciphertext-policy hierarchical attribute-based encryption with short
  ciphertexts.
\newblock {\em Information Sciences 275\/} (2014), 370--384.

\bibitem{do2020traceable}
{\sc Do, X.~T., Phan, D.~H., and Pointcheval, D.}
\newblock Traceable inner product functional encryption.
\newblock In {\em Cryptographers’ Track at the RSA Conference\/} (2020),
  Springer, pp.~564--585.

\bibitem{dufour2019unbounded}
{\sc Dufour-Sans, E., and Pointcheval, D.}
\newblock Unbounded inner-product functional encryption with succinct keys.
\newblock In {\em International Conference on Applied Cryptography and Network
  Security\/} (2019), Springer, pp.~426--441.

\bibitem{escala2017algebraic}
{\sc Escala, A., Herold, G., Kiltz, E., Rafols, C., and Villar, J.}
\newblock An algebraic framework for {Diffie}--{Hellman} assumptions.
\newblock {\em Journal of cryptology 30}, 1 (2017), 242--288.

\bibitem{FT15}
{\sc Fan, C.-I., and Tseng, Y.-F.}
\newblock Anonymous multi-receiver identity-based authenticated encryption with
  {CCA} security.
\newblock {\em Symmetry 7}, 4 (2015), 1856--1881.

\bibitem{FLR+10}
{\sc Fischlin, M., Lehmann, A., Ristenpart, T., Shrimpton, T., Stam, M., and
  Tessaro, S.}
\newblock Random oracles with (out) programmability.
\newblock In {\em International Conference on the Theory and Application of
  Cryptology and Information Security\/} (2010), Springer, pp.~303--320.

\bibitem{garey1979computers}
{\sc Garey, M.~R., and Johnson, D.~S.}
\newblock {\em Computers and Intractability: {A} Guide to the Theory of
  {NP}-Completeness}.
\newblock W. H. Freeman \& Co., USA, 1979.

\bibitem{garg2013candidate}
{\sc Garg, S., Gentry, C., and Halevi, S.}
\newblock Candidate multilinear maps from ideal lattices.
\newblock In {\em Annual International Conference on the Theory and
  Applications of Cryptographic Techniques\/} (2013), Springer, pp.~1--17.

\bibitem{GGH+16}
{\sc Garg, S., Gentry, C., Halevi, S., Raykova, M., Sahai, A., and Waters, B.}
\newblock Candidate indistinguishability obfuscation and functional encryption
  for all circuits.
\newblock {\em SIAM Journal on Computing 45}, 3 (2016), 882--929.

\bibitem{GGHZ16}
{\sc Garg, S., Gentry, C., Halevi, S., and Zhandry, M.}
\newblock Functional encryption without obfuscation.
\newblock In {\em Theory of Cryptography Conference\/} (2016), Springer,
  pp.~480--511.

\bibitem{Gay19}
{\sc Gay, R.}
\newblock {\em Public-Key Encryption, Revisited: {Tight} Security and Richer
  Functionalities}.
\newblock PhD thesis, PSL Research University, 2019.

\bibitem{gay2020new}
{\sc Gay, R.}
\newblock A new paradigm for public-key functional encryption for degree-2
  polynomials.
\newblock In {\em IACR International Conference on Public-Key Cryptography\/}
  (2020), Springer, pp.~95--120.

\bibitem{Gen06}
{\sc Gentry, C.}
\newblock Practical identity-based encryption without random oracles.
\newblock In {\em Annual international conference on the theory and
  applications of cryptographic techniques\/} (2006), Springer, pp.~445--464.

\bibitem{Gen09}
{\sc Gentry, C.}
\newblock Fully homomorphic encryption using ideal lattices.
\newblock In {\em Proceedings of the forty-first annual ACM symposium on Theory
  of computing\/} (2009), pp.~169--178.

\bibitem{gentry2015graph}
{\sc Gentry, C., Gorbunov, S., and Halevi, S.}
\newblock Graph-induced multilinear maps from lattices.
\newblock In {\em Theory of Cryptography Conference\/} (2015), Springer,
  pp.~498--527.

\bibitem{GS02}
{\sc Gentry, C., and Silverberg, A.}
\newblock Hierarchical {ID}-based cryptography.
\newblock In {\em International conference on the theory and application of
  cryptology and information security\/} (2002), Springer, pp.~548--566.

\bibitem{GAS16}
{\sc Giacon, F., Aragona, R., and Sala, M.}
\newblock A proof of security for a key-policy {RS}-{ABE} scheme.
\newblock {\em JP Journal of Algebra, Number Theory and Applications 40}, 1
  (2018), 29--90.

\bibitem{GGG+14}
{\sc Goldwasser, S., Gordon, S.~D., Goyal, V., Jain, A., Katz, J., Liu, F.-H.,
  Sahai, A., Shi, E., and Zhou, H.-S.}
\newblock Multi-input functional encryption.
\newblock In {\em Annual International Conference on the Theory and
  Applications of Cryptographic Techniques\/} (2014), Springer, pp.~578--602.

\bibitem{GKP+13}
{\sc Goldwasser, S., Kalai, Y., Popa, R.~A., Vaikuntanathan, V., and Zeldovich,
  N.}
\newblock Reusable garbled circuits and succinct functional encryption.
\newblock In {\em Proceedings of the forty-fifth annual ACM symposium on Theory
  of computing\/} (2013), pp.~555--564.

\bibitem{GVW12}
{\sc Gorbunov, S., Vaikuntanathan, V., and Wee, H.}
\newblock Functional encryption with bounded collusions via multi-party
  computation.
\newblock In {\em Annual Cryptology Conference\/} (2012), Springer,
  pp.~162--179.

\bibitem{GVW13}
{\sc Gorbunov, S., Vaikuntanathan, V., and Wee, H.}
\newblock Attribute-based encryption for circuits.
\newblock {\em Journal of the ACM (JACM) 62}, 6 (2015), 1--33.

\bibitem{GVW15}
{\sc Gorbunov, S., Vaikuntanathan, V., and Wee, H.}
\newblock Predicate encryption for circuits from {LWE}.
\newblock In {\em Annual Cryptology Conference\/} (2015), Springer,
  pp.~503--523.

\bibitem{goyal2007reducing}
{\sc Goyal, V.}
\newblock Reducing trust in the {PKG} in identity based cryptosystems.
\newblock In {\em Annual International Cryptology Conference\/} (2007),
  Springer, pp.~430--447.

\bibitem{GPSW06}
{\sc Goyal, V., Pandey, O., Sahai, A., and Waters, B.}
\newblock Attribute-based encryption for fine-grained access control of
  encrypted data.
\newblock In {\em Proceedings of the 13th ACM conference on Computer and
  communications security\/} (2006), pp.~89--98.

\bibitem{HK+20}
{\sc Hanaoka, G., Komatsu, M., Ohara, K., Sakai, Y., and Yamada, S.}
\newblock Semantic definition of anonymity in identity-based encryption and its
  relation to indistinguishability-based definition.
\newblock In {\em European Symposium on Research in Computer Security\/}
  (2020), Springer, pp.~65--85.

\bibitem{HY18}
{\sc Hanaoka, G., and Yamada, S.}
\newblock A survey on identity-based encryption from lattices.
\newblock In {\em Mathematical modelling for next-generation cryptography}.
  Springer, 2018, pp.~349--365.

\bibitem{HW+16}
{\sc He, K., Weng, J., Liu, J.-N., Liu, J.~K., Liu, W., and Deng, R.~H.}
\newblock Anonymous identity-based broadcast encryption with chosen-ciphertext
  security.
\newblock In {\em Proceedings of the 11th ACM on Asia Conference on Computer
  and Communications Security\/} (2016), pp.~247--255.

\bibitem{HL02}
{\sc Horwitz, J., and Lynn, B.}
\newblock Toward hierarchical identity-based encryption.
\newblock In {\em International conference on the theory and applications of
  cryptographic techniques\/} (2002), Springer, pp.~466--481.

\bibitem{HN10}
{\sc Hur, J., and Noh, D.~K.}
\newblock Attribute-based access control with efficient revocation in data
  outsourcing systems.
\newblock {\em IEEE Transactions on Parallel and Distributed Systems 22}, 7
  (2010), 1214--1221.

\bibitem{IP08}
{\sc Iovino, V., and Persiano, G.}
\newblock Hidden-vector encryption with groups of prime order.
\newblock In {\em International Conference on Pairing-Based Cryptography\/}
  (2008), Springer, pp.~75--88.

\bibitem{KSW08}
{\sc Katz, J., Sahai, A., and Waters, B.}
\newblock Predicate encryption supporting disjunctions, polynomial equations,
  and inner products.
\newblock In {\em annual international conference on the theory and
  applications of cryptographic techniques\/} (2008), Springer, pp.~146--162.

\bibitem{kim2019new}
{\sc Kim, S., Kim, J., and Seo, J.~H.}
\newblock A new approach to practical function-private inner product
  encryption.
\newblock {\em Theoretical Computer Science 783\/} (2019), 22--40.

\bibitem{kim2018function}
{\sc Kim, S., Lewi, K., Mandal, A., Montgomery, H., Roy, A., and Wu, D.~J.}
\newblock Function-hiding inner product encryption is practical.
\newblock In {\em International Conference on Security and Cryptography for
  Networks\/} (2018), Springer, pp.~544--562.

\bibitem{LCH13}
{\sc Lee, C.-C., Chung, P.-S., and Hwang, M.-S.}
\newblock A survey on attribute-based encryption schemes of access control in
  cloud environments.
\newblock {\em IJ Network Security 15}, 4 (2013), 231--240.

\bibitem{LC+13}
{\sc Lee, K., Choi, S.~G., Lee, D.~H., Park, J.~H., and Yung, M.}
\newblock Self-updatable encryption: {Time} constrained access control with
  hidden attributes and better efficiency.
\newblock In {\em International Conference on the Theory and Application of
  Cryptology and Information Security\/} (2013), Springer, pp.~235--254.

\bibitem{LOS+10}
{\sc Lewko, A., Okamoto, T., Sahai, A., Takashima, K., and Waters, B.}
\newblock Fully secure functional encryption: {Attribute}-based encryption and
  (hierarchical) inner product encryption.
\newblock In {\em Annual International Conference on the Theory and
  Applications of Cryptographic Techniques\/} (2010), Springer, pp.~62--91.

\bibitem{lewko2010new}
{\sc Lewko, A., and Waters, B.}
\newblock New techniques for dual system encryption and fully secure {HIBE}
  with short ciphertexts.
\newblock In {\em Theory of Cryptography Conference\/} (2010), Springer,
  pp.~455--479.

\bibitem{lewko2011decentralizing}
{\sc Lewko, A., and Waters, B.}
\newblock Decentralizing attribute-based encryption.
\newblock In {\em Annual international conference on the theory and
  applications of cryptographic techniques\/} (2011), Springer, pp.~568--588.

\bibitem{lewko2011unbounded}
{\sc Lewko, A., and Waters, B.}
\newblock Unbounded {HIBE} and attribute-based encryption.
\newblock In {\em Annual International Conference on the Theory and
  Applications of Cryptographic Techniques\/} (2011), Springer, pp.~547--567.

\bibitem{li2020decentralized}
{\sc Li, J., Hu, S., Zhang, Y., and Han, J.}
\newblock A decentralized multi-authority ciphertext-policy attribute-based
  encryption with mediated obfuscation.
\newblock {\em Soft Computing 24}, 3 (2020), 1869--1882.

\bibitem{LYZ19}
{\sc Li, J., Yu, Q., and Zhang, Y.}
\newblock Hierarchical attribute based encryption with continuous
  leakage-resilience.
\newblock {\em Information Sciences 484\/} (2019), 113--134.

\bibitem{LT19}
{\sc Libert, B., and {\c{T}}i{\c{t}}iu, R.}
\newblock Multi-client functional encryption for linear functions in the
  standard model from {LWE}.
\newblock In {\em International Conference on the Theory and Application of
  Cryptology and Information Security\/} (2019), Springer, pp.~520--551.

\bibitem{lin2017indistinguishability}
{\sc Lin, H.}
\newblock Indistinguishability obfuscation from {SXDH} on 5-linear maps and
  locality-5 {PRGs}.
\newblock In {\em Annual International Cryptology Conference\/} (2017),
  Springer, pp.~599--629.

\bibitem{lin2010secure}
{\sc Lin, H., Cao, Z., Liang, X., and Shao, J.}
\newblock Secure threshold multi authority attribute based encryption without a
  central authority.
\newblock {\em Information Sciences 180}, 13 (2010), 2618--2632.

\bibitem{LT17}
{\sc Lin, H., and Tessaro, S.}
\newblock Indistinguishability obfuscation from trilinear maps and block-wise
  local {PRGs}.
\newblock In {\em Annual International Cryptology Conference\/} (2017),
  Springer, pp.~630--660.

\bibitem{liu2021efficient}
{\sc Liu, W., Huang, Q., Chen, X., and Li, H.}
\newblock Efficient functional encryption for inner product with
  simulation-based security.
\newblock {\em Cybersecurity 4}, 1 (2021), 1--13.

\bibitem{liu2013blackbox}
{\sc Liu, Z., Cao, Z., and Wong, D.~S.}
\newblock Blackbox traceable {CP-ABE}: {How} to catch people leaking their keys
  by selling decryption devices on ebay.
\newblock In {\em Proceedings of the 2013 ACM SIGSAC conference on Computer \&
  communications security\/} (2013), pp.~475--486.

\bibitem{liu2015practical}
{\sc Liu, Z., and Wong, D.~S.}
\newblock Practical ciphertext-policy attribute-based encryption: {Traitor}
  tracing, revocation, and large universe.
\newblock In {\em International Conference on Applied Cryptography and Network
  Security\/} (2015), Springer, pp.~127--146.

\bibitem{liu2015traceable}
{\sc Liu, Z., and Wong, D.~S.}
\newblock Traceable {CP-ABE} on prime order groups: {Fully} secure and fully
  collusion-resistant blackbox traceable.
\newblock In {\em International Conference on Information and Communications
  Security\/} (2015), Springer, pp.~109--124.

\bibitem{LMS15}
{\sc Longo, R., Marcolla, C., and Sala, M.}
\newblock Key-policy multi-authority attribute-based encryption.
\newblock In {\em International Conference on Algebraic Informatics\/} (2015),
  Springer, pp.~152--164.

\bibitem{LMS16}
{\sc Longo, R., Marcolla, C., and Sala, M.}
\newblock Collaborative multi-authority {KP}-{ABE} for shorter keys and
  parameters.
\newblock In {\em International Conference on Algebraic Informatics\/} (2017),
  p.~67.
\newblock \url{https://eprint.iacr.org/2016/262}.

\bibitem{ma2018mmap}
{\sc Ma, F., and Zhandry, M.}
\newblock The {MMap} strikes back: {Obfuscation} and new multilinear maps
  immune to {CLT13} zeroizing attacks.
\newblock In {\em Theory of Cryptography Conference\/} (2018), Springer,
  pp.~513--543.

\bibitem{MWL18}
{\sc Ma, X., Wang, X., and Lin, D.}
\newblock Anonymous identity-based encryption with identity recovery.
\newblock In {\em Australasian Conference on Information Security and
  Privacy\/} (2018), Springer, pp.~360--375.

\bibitem{MOV93}
{\sc Menezes, A.~J., Okamoto, T., and Vanstone, S.~A.}
\newblock Reducing elliptic curve logarithms to logarithms in a finite field.
\newblock {\em iEEE Transactions on information Theory 39}, 5 (1993),
  1639--1646.

\bibitem{miller2004weil}
{\sc Miller, V.~S.}
\newblock The {Weil} pairing, and its efficient calculation.
\newblock {\em Journal of cryptology 17}, 4 (2004), 235--261.

\bibitem{MNT01}
{\sc Miyaji, A., Nakabayashi, M., and Takano, S.}
\newblock New explicit conditions of elliptic curve traces for {FR}-reduction.
\newblock {\em IEICE transactions on fundamentals of electronics,
  communications and computer sciences 84}, 5 (2001), 1234--1243.

\bibitem{muller2008distributed}
{\sc M{\"u}ller, S., Katzenbeisser, S., and Eckert, C.}
\newblock Distributed attribute-based encryption.
\newblock In {\em International Conference on Information Security and
  Cryptology\/} (2008), Springer, pp.~20--36.

\bibitem{Nac07}
{\sc Naccache, D.}
\newblock Secure and practical identity-based encryption.
\newblock {\em IET Information Security 1}, 2 (2007), 59--64.

\bibitem{naveed2014controlled}
{\sc Naveed, M., Agrawal, S., Prabhakaran, M., Wang, X., Ayday, E., Hubaux,
  J.-P., and Gunter, C.}
\newblock Controlled functional encryption.
\newblock In {\em Proceedings of the 2014 ACM SIGSAC Conference on Computer and
  Communications Security\/} (2014), pp.~1280--1291.

\bibitem{ning2014large}
{\sc Ning, J., Cao, Z., Dong, X., Wei, L., and Lin, X.}
\newblock Large universe ciphertext-policy attribute-based encryption with
  white-box traceability.
\newblock In {\em European Symposium on Research in Computer Security\/}
  (2014), Springer, pp.~55--72.

\bibitem{okamoto2008homomorphic}
{\sc Okamoto, T., and Takashima, K.}
\newblock Homomorphic encryption and signatures from vector decomposition.
\newblock In {\em International Conference on Pairing-Based Cryptography\/}
  (2008), Springer, pp.~57--74.

\bibitem{OT09}
{\sc Okamoto, T., and Takashima, K.}
\newblock Hierarchical predicate encryption for inner-products.
\newblock In {\em International Conference on the Theory and Application of
  Cryptology and Information Security\/} (2009), Springer, pp.~214--231.

\bibitem{OT10}
{\sc Okamoto, T., and Takashima, K.}
\newblock Fully secure functional encryption with general relations from the
  decisional linear assumption.
\newblock In {\em Annual cryptology conference\/} (2010), Springer,
  pp.~191--208.

\bibitem{OT12}
{\sc Okamoto, T., and Takashima, K.}
\newblock Adaptively attribute-hiding (hierarchical) inner product encryption.
\newblock In {\em Annual International Conference on the Theory and
  Applications of Cryptographic Techniques\/} (2012), Springer, pp.~591--608.

\bibitem{okamoto2012fully}
{\sc Okamoto, T., and Takashima, K.}
\newblock Fully secure unbounded inner-product and attribute-based encryption.
\newblock In {\em International Conference on the Theory and Application of
  Cryptology and Information Security\/} (2012), Springer, pp.~349--366.

\bibitem{okamoto2013efficient}
{\sc Okamoto, T., and Takashima, K.}
\newblock Efficient (hierarchical) inner-product encryption tightly reduced
  from the decisional linear assumption.
\newblock {\em IEICE Transactions on Fundamentals of Electronics,
  Communications and Computer Sciences 96}, 1 (2013), 42--52.

\bibitem{ON10}
{\sc O'Neill, A.}
\newblock Definitional issues in functional encryption.
\newblock Cryptology ePrint Archive, Report 2009/556, 2010.
\newblock \url{https://eprint.iacr.org/2010/556}.

\bibitem{OSW07}
{\sc Ostrovsky, R., Sahai, A., and Waters, B.}
\newblock Attribute-based encryption with non-monotonic access structures.
\newblock In {\em Proceedings of the 14th ACM conference on Computer and
  communications security\/} (2007), pp.~195--203.

\bibitem{page2006comparison}
{\sc Page, D., Smart, N.~P., and Vercauteren, F.}
\newblock A comparison of {MNT} curves and supersingular curves.
\newblock {\em Applicable Algebra in Engineering, Communication and Computing
  17}, 5 (2006), 379--392.

\bibitem{Par10}
{\sc Park, J.~H.}
\newblock Efficient hidden vector encryption for conjunctive queries on
  encrypted data.
\newblock {\em IEEE Transactions on Knowledge and Data Engineering 23}, 10
  (2010), 1483--1497.

\bibitem{Par11}
{\sc Park, J.~H.}
\newblock Inner-product encryption under standard assumptions.
\newblock {\em Designs, Codes and Cryptography 58}, 3 (2011), 235--257.

\bibitem{PL+13}
{\sc Park, J.~H., Lee, K., Susilo, W., and Lee, D.~H.}
\newblock Fully secure hidden vector encryption under standard assumptions.
\newblock {\em Information Sciences 232\/} (2013), 188--207.

\bibitem{patarin1997trapdoor}
{\sc Patarin, J., and Goubin, L.}
\newblock Trapdoor one-way permutations and multivariate polynomials.
\newblock In {\em International Conference on Information and Communications
  Security\/} (1997), Springer, pp.~356--368.

\bibitem{Pei09}
{\sc Peikert, C.}
\newblock Bonsai trees (or, arboriculture in lattice-based cryptography).
\newblock Cryptology ePrint Archive, Report 2009/359, 2009.
\newblock \url{https://eprint.iacr.org/2009/359}.

\bibitem{phan2011identity}
{\sc Phan, D.~H., and Trinh, V.~C.}
\newblock Identity-based trace and revoke schemes.
\newblock In {\em International Conference on Provable Security\/} (2011),
  Springer, pp.~204--221.

\bibitem{QLDJ14}
{\sc Qiao, Z., Liang, S., Davis, S., and Jiang, H.}
\newblock Survey of attribute based encryption.
\newblock In {\em 15th IEEE/ACIS International Conference on Software
  Engineering, Artificial Intelligence, Networking and Parallel/Distributed
  Computing (SNPD)\/} (2014), IEEE, pp.~1--6.

\bibitem{rahulamathavan2015user}
{\sc Rahulamathavan, Y., Veluru, S., Han, J., Li, F., Rajarajan, M., and Lu,
  R.}
\newblock User collusion avoidance scheme for privacy-preserving decentralized
  key-policy attribute-based encryption.
\newblock {\em IEEE Transactions on Computers 65}, 9 (2015), 2939--2946.

\bibitem{Reg09}
{\sc Regev, O.}
\newblock On lattices, learning with errors, random linear codes, and
  cryptography.
\newblock {\em Journal of the ACM (JACM) 56}, 6 (2009), 1--40.

\bibitem{debnath2020post}
{\sc S.~K.~Debnath, S.~Mesnager, K.~D., and Kundu, N.}
\newblock Post-quantum secure inner product functional encryption using
  multivariate public key cryptography.
\newblock {\em Journal Mediterranean Journal of Mathematics\/} (To appear).

\bibitem{SS10}
{\sc Sahai, A., and Seyalioglu, H.}
\newblock Worry-free encryption: {Functional} encryption with public keys.
\newblock In {\em Proceedings of the 17th ACM conference on Computer and
  communications security\/} (2010), pp.~463--472.

\bibitem{SW05}
{\sc Sahai, A., and Waters, B.}
\newblock Fuzzy identity-based encryption.
\newblock In {\em Annual international conference on the theory and
  applications of cryptographic techniques\/} (2005), Springer, pp.~457--473.

\bibitem{SV+10}
{\sc Sedghi, S., Van~Liesdonk, P., Nikova, S., Hartel, P., and Jonker, W.}
\newblock Searching keywords with wildcards on encrypted data.
\newblock In {\em International Conference on Security and Cryptography for
  Networks\/} (2010), Springer, pp.~138--153.

\bibitem{shamir1979share}
{\sc Shamir, A.}
\newblock How to share a secret.
\newblock {\em Communications of the ACM 22}, 11 (1979), 612--613.

\bibitem{Sha84}
{\sc Shamir, A.}
\newblock Identity-based cryptosystems and signature schemes.
\newblock In {\em Workshop on the theory and application of cryptographic
  techniques\/} (1984), Springer, pp.~47--53.

\bibitem{shen2009predicate}
{\sc Shen, E., Shi, E., and Waters, B.}
\newblock Predicate privacy in encryption systems.
\newblock In {\em Theory of Cryptography Conference\/} (2009), Springer,
  pp.~457--473.

\bibitem{SW08}
{\sc Shi, E., and Waters, B.}
\newblock Delegating capabilities in predicate encryption systems.
\newblock In {\em International Colloquium on Automata, Languages, and
  Programming\/} (2008), Springer, pp.~560--578.

\bibitem{shor1999polynomial}
{\sc Shor, P.~W.}
\newblock Polynomial-time algorithms for prime factorization and discrete
  logarithms on a quantum computer.
\newblock {\em SIAM review 41}, 2 (1999), 303--332.

\bibitem{Sil09}
{\sc Silverman, J.~H.}
\newblock {\em The arithmetic of elliptic curves}, vol.~106.
\newblock Springer Science \& Business Media, 2009.

\bibitem{soroush2020verifiable}
{\sc Soroush, N., Iovino, V., Rial, A., Roenne, P.~B., and Ryan, P.~Y.}
\newblock Verifiable inner product encryption scheme.
\newblock In {\em IACR International Conference on Public-Key Cryptography\/}
  (2020), Springer, pp.~65--94.

\bibitem{takashima2008efficiently}
{\sc Takashima, K.}
\newblock Efficiently computable distortion maps for supersingular curves.
\newblock In {\em International Algorithmic Number Theory Symposium\/} (2008),
  Springer, pp.~88--101.

\bibitem{tang2004non}
{\sc Tang, C., Pei, D., Liu, Z., and He, Y.}
\newblock Non-interactive and information-theoretic secure publicly verifiable
  secret sharing.
\newblock Cryptology ePrint Archive, Report 2004/201, 2004.
\newblock \url{https://eprint.iacr.org/2004/201}.

\bibitem{tang2010verifiable}
{\sc Tang, Q., and Ji, D.}
\newblock Verifiable attribute-based encryption.
\newblock {\em IJ Network Security 10}, 2 (2010), 114--120.

\bibitem{TAA19}
{\sc Tea, B.~C., Ariffin, M. R.~K., and Asbullah, M.~A.}
\newblock Identity-based encryption schemes--{A} review.
\newblock {\em Journal of Multidisciplinary Engineering Science and Technology
  (JMEST) 6}, 12 (2019).

\bibitem{tomida2016efficient}
{\sc Tomida, J., Abe, M., and Okamoto, T.}
\newblock Efficient functional encryption for inner-product values with
  full-hiding security.
\newblock In {\em International Conference on Information Security\/} (2016),
  Springer, pp.~408--425.

\bibitem{tomida2020unbounded}
{\sc Tomida, J., and Takashima, K.}
\newblock Unbounded inner product functional encryption from bilinear maps.
\newblock {\em Japan Journal of Industrial and Applied Mathematics 37}, 3
  (2020), 723--779.

\bibitem{van2020multi}
{\sc van~de Kamp, T., Peter, A., and Jonker, W.}
\newblock A multi-authority approach to various predicate encryption types.
\newblock {\em Designs, codes and cryptography 88}, 2 (2020), 363--390.

\bibitem{WLW10}
{\sc Wang, G., Liu, Q., and Wu, J.}
\newblock Hierarchical attribute-based encryption for fine-grained access
  control in cloud storage services.
\newblock In {\em Proceedings of the 17th ACM conference on Computer and
  communications security\/} (2010), pp.~735--737.

\bibitem{Wat05}
{\sc Waters, B.}
\newblock Efficient identity-based encryption without random oracles.
\newblock In {\em Annual International Conference on the Theory and
  Applications of Cryptographic Techniques\/} (2005), Springer, pp.~114--127.

\bibitem{Wat09}
{\sc Waters, B.}
\newblock Dual system encryption: {Realizing} fully secure {IBE} and {HIBE}
  under simple assumptions.
\newblock In {\em Annual International Cryptology Conference\/} (2009),
  Springer, pp.~619--636.

\bibitem{Wat11}
{\sc Waters, B.}
\newblock Ciphertext-policy attribute-based encryption: {An} expressive,
  efficient, and provably secure realization.
\newblock In {\em International Workshop on Public Key Cryptography\/} (2011),
  Springer, pp.~53--70.

\bibitem{Wee17}
{\sc Wee, H.}
\newblock Attribute-hiding predicate encryption in bilinear groups, revisited.
\newblock In {\em Theory of Cryptography Conference\/} (2017), Springer,
  pp.~206--233.

\bibitem{Wee20}
{\sc Wee, H.}
\newblock Functional encryption for quadratic functions from $k$-{Lin},
  revisited.
\newblock In {\em Theory of Cryptography Conference\/} (2020), Springer,
  pp.~210--228.

\bibitem{XL+16}
{\sc Xu, P., Li, J., Wang, W., and Jin, H.}
\newblock Anonymous identity-based broadcast encryption with constant
  decryption complexity and strong security.
\newblock In {\em Proceedings of the 11th ACM on Asia Conference on Computer
  and Communications Security\/} (2016), pp.~223--233.

\bibitem{yang2018improving}
{\sc Yang, Y., Chen, X., Chen, H., and Du, X.}
\newblock Improving privacy and security in decentralizing multi-authority
  attribute-based encryption in cloud computing.
\newblock {\em IEEE Access 6\/} (2018), 18009--18021.

\bibitem{ZYT13}
{\sc Zhang, M., Yang, B., and Takagi, T.}
\newblock Bounded leakage-resilient functional encryption with hidden vector
  predicate.
\newblock {\em The Computer Journal 56}, 4 (2013), 464--477.

\bibitem{ZD+20}
{\sc Zhang, Y., Deng, R.~H., Xu, S., Sun, J., Li, Q., and Zheng, D.}
\newblock Attribute-based encryption for cloud computing access control: {A}
  survey.
\newblock {\em ACM Computing Surveys (CSUR) 53}, 4 (2020), 1--41.

\bibitem{zhao2018improved}
{\sc Zhao, Q., Zeng, Q., and Liu, X.}
\newblock Improved construction for inner product functional encryption.
\newblock {\em Security and Communication Networks 2018\/} (2018).

\bibitem{zhao2018simulation}
{\sc Zhao, Q., Zeng, Q., Liu, X., and Xu, H.}
\newblock Simulation-based security of function-hiding inner product
  encryption.
\newblock {\em Science China Information Sciences 61}, 4 (2018), 1--3.

\end{thebibliography}

\end{document}